\documentclass[3p,times]{elsarticle} 
\makeatletter
\def\ps@pprintTitle{%
  \let\@oddhead\@empty
  \let\@evenhead\@empty
  \def\@oddfoot{%
    \makebox[\textwidth][s]{
      \hfill
      \makebox[0pt][c]{\thepage}
      \hfill
      \small{\textit{\today}}
    }%
  }%
  \let\@evenfoot\@oddfoot}
\makeatother
\usepackage{amsthm, amsmath, amssymb, amsfonts, graphicx, epsfig}
\usepackage{accents}

\usepackage{bbm}
\usepackage{mathrsfs}

\usepackage{epstopdf}

\usepackage{graphicx}
\usepackage[font=sl,labelfont=bf]{caption}
\usepackage{subcaption}

\usepackage{enumerate}
\usepackage[colorlinks=true, pdfstartview=FitV, linkcolor=blue,
            citecolor=blue, urlcolor=blue]{hyperref}
\usepackage[usenames]{color}
\usepackage{multirow}

\usepackage{url}
\makeatletter\def\url@leostyle{%
 \@ifundefined{selectfont}{\def\UrlFont{\sf}}{\def\UrlFont{\scriptsize\ttfamily}}} \makeatother

\newtheorem{theorem}{Theorem}
\newtheorem{proposition}[theorem]{Proposition}

\theoremstyle{definition}
\newtheorem{definition}[theorem]{Definition}
\newtheorem{example}[theorem]{Example}
\theoremstyle{remark}
\newtheorem{remark}[theorem]{Remark}

\usepackage{mdframed}

\theoremstyle{definition}
\newmdtheoremenv{definition2}[theorem]{Definition}
\newmdtheoremenv{proposition2}[theorem]{Proposition}
\newmdtheoremenv{theorem2}[theorem]{Theorem}
\newmdtheoremenv{corollary2}[theorem]{Corollary}

\numberwithin{equation}{section}
\numberwithin{theorem}{section}

\definecolor{Red}{rgb}{0.9,0,0.0}
\definecolor{Blue}{rgb}{0,0.0,1.0}

\newcommand{\gray}[1]{\begin{color}[rgb]{0.5, 0.5, 0.5}{#1}\end{color}}

\def\cM{\mathcal{M}}

\def\cX{\mathcal{X}}

\def\bE{\mathbb{E}}

\def\bN{\mathbb{N}}

\def\bP{\mathbb{P}}

\def\bR{\mathbb{R}}

\def\sG{\mathscr{G}}

\def\bfU{\mathbf{U}}
\def\bfV{\mathbf{V}}

\def\bfX{\mathbf{X}}
\def\bfY{\mathbf{Y}}
\def\bfZ{\mathbf{Z}}

\def\bfx{\mathbf{x}}
\def\bfy{\mathbf{y}}

\newcommand{\1}{\mathbbm{1}}            
\newcommand{\set}[1]{\{#1\}}            
\renewcommand{\mid}{\;|\;}              
\newcommand{\dif}{\operatorname{d}\!}   

\DeclareMathOperator*{\argmin}{arg\,min} 

\DeclareMathOperator{\ExpVAR}{ExpVaR}          

\DeclareMathOperator{\var}{\mathrm{VaR}}           
\DeclareMathOperator{\ES}{\mathrm{ES}}                    
\DeclareMathOperator{\wvar}{\mathrm{WVaR}}         
     
\makeatletter
\def\namedlabel#1#2{\begingroup
    #2%
    \def\@currentlabel{#2}%
    \phantomsection\label{#1}\endgroup
}
\makeatother

\biboptions{authoryear}

\journal{TBD}

\usepackage{setspace}
\usepackage{booktabs}
\usepackage{multirow}
\usepackage{lipsum}
\usepackage{makecell}
\usepackage{ulem}

\usepackage{endnotes}

\let\footnote=\endnote

\long\def\rev#1{{\color{black}#1}}

\begin{document}

\begin{frontmatter}

\title{Coherent estimation of risk measures}


\author[label1]{Martin Aichele}\ead{martin.aichele@ecb.europa.eu}
\author[label2]{Igor Cialenco}\ead{cialenco@illinoistech.edu}
\author[label3]{Damian Jelito\corref{cor1}}\ead{damian.jelito@uj.edu.pl}
\author[label3]{Marcin Pitera}\ead{marcin.pitera@uj.edu.pl}
\address[label1]{European Central Bank, Sonnemannstra\ss e 22, 60314 Frankfurt am Main, Germany} 
\address[label2]{Department of Applied Mathematics, Illinois Institute of Technology, 10 W 32nd Str, Building REC, Room 220, Chicago, IL 60616, USA}
\address[label3]{Institute of Mathematics, Jagiellonian University, S. {\L}ojasiewicza 6, 30-348 Krak{\'o}w, Poland}
   \cortext[cor1]{Corresponding author}

\begin{abstract}
We develop a statistical framework for risk estimation, inspired by the axiomatic theory of risk measures. Coherent risk estimators---functionals of P\&L samples inheriting the economic properties of risk measures---are defined and characterized through robust representations linked to $L$-estimators. The framework provides a canonical methodology for constructing estimators with sound financial and statistical properties, unifying risk measure theory, principles for capital adequacy, and practical statistical challenges in market risk. \rev{Numerical illustrations based on simulated and market data demonstrate that coherence of a risk measure does not necessarily carry over to its estimators and show that alternative admissible weight structures within the CRE representation can lead to substantially different capital adequacy outcomes.}
\end{abstract}


\end{frontmatter}




\section{Introduction}
\rev{Managing risk is central to any financial institution, whether driven by regulatory requirements or internal monitoring. Equally, financial regulatory bodies are mandated to ensure that institutions remain solvent under adverse scenarios through appropriate risk assessment frameworks.} In either case, the fundamental problem is to design adequate risk measurement tools that can capture the \rev{usually highly complex} financial risk profiles based on limited~data.

The existing risk measurement methodologies, broadly speaking, have evolved along the following pathway.  In the first step, we {\it design a risk measure}, say $\rho$, under the assumption that the true law of the future's profit and loss vector of a financial position (P\&L), say $X$, is known or can be found. The function $\rho$ maps the random variable $X$ to a real number $\rho(X)$, indicating how risky the underlying position is. In the second step, we {\it estimate the risk of a financial position}, say $\hat\rho(X)$, assuming that the true distribution of $X$ is not known, and only a finite sample of $X$ is available.  Among risk measures that are often used for the first step, one can mention the value at risk ($\var$) used as the primary metric for capital requirements under the Basel II market risk framework, or the expected shortfall ($\ES$), adopted by the Basel Committee on Banking Supervision (BCBS) as part of the Basel III reforms, see \cite{Bas2006,Bas2013,Bas2019} for \rev{regulatory} details \rev{and Section~\ref{sec:preliminaries} for definitions}. 
For the second step, there are many well-known risk estimation frameworks linked, e.g., to historical simulation or Monte Carlo methods. We refer to \cite{McNeilFreEmb2015} and \cite{Car2009} for other examples of risk measures and an overview of the most popular estimation approaches.

In this work, we focus on the second step and adopt a novel perspective, fundamentally different from the existing literature, focused on the estimation function of $\rho$, say $\hat\rho$, which is later used to estimate $\rho(X)$. We argue that, similar to the risk measures themselves, any estimation must also satisfy a set of desirable financial normative properties, postulated a priori, and in addition be a `good approximation'\footnote{There is a subtle distinction between estimation and approximation. Estimation refers to inferring an unknown true value from limited data, with an emphasis on its statistical properties. Approximation, in contrast, involves simplifying a known value or formula to make it computationally tractable, while analyzing the resulting numerical error. Our approach in this work combines elements of both, but for consistency we refer to them jointly as estimation.}. This approach is motivated by the recognition that risk quantification procedures serve primarily to determine capital reserves for mitigating \rev{risk of} exposures, rather than to provide mere approximations of intrinsic risk values. To this end, we introduce the notion of the \textit{coherent risk estimator} (CRE) that maps 
\rev{P\&L samples to real numbers and satisfies specific properties that stem from} the axiomatic risk measure framework of \cite{ArtDelEbeHea1997,ArtDelEbeHea1999} that laid the foundation for the modern theory of risk measures. 
\rev{The chosen} properties encode financially and economically meaningful requirements--such as monotonicity with respect to losses to allow ordering of positions, cash-additivity to reflect a minimum  capital requirement, positive homogeneity to capture the proportional scaling of risk in a rescaled portfolio, or subadditivity to account for diversification benefits; see Section~\ref{sec:preliminaries} for precise definitions and further discussion. 

\rev{The key novel contribution of this paper is the development of the robust representation and structural classification of all CREs, see Theorem~\ref{th:coherent_estimator_representation}. The developed estimator-level dual representation reveals an inherent link between CREs and suprema of linear estimators, thereby bridging risk management and robust statistics. Indeed, we show that a CRE is} law-invariant\footnote{Similar to risk measures, a law-invariant CRE does not depend on the ordering of the input sample. See Definition~\ref{def:law_invariant_estimator}.} if and only if it can be represented as a supremum over a set of $L$-estimators; see Theorem~\ref{th:law_inv_est_representation}. Moreover, assuming additionally that a CRE is \rev{comonotonic},
we show that it can be represented as \rev{a single} $L$-estimator; see Theorem~\ref{th:comonotonic_CRE}. The importance of such results is evident.

\rev{Our canonical framework not only allows the construction of risk estimators that are designed to possess the desired risk management properties but enables the systematic selection of favorable estimators from the wide variety of existing alternatives. It also ensures that the chosen methods remain practically applicable, are robust, distribution-free, and aligned with supervisory expectations, see e.g.~\cite{EGIM}. As a byproduct, we show that many popular parametric and semi-parametric risk estimators are not coherent as their structure differs from that of CREs.}

\rev{Furthermore, through a series of examples, we show that the coherence of the underlying theoretical risk measure is not necessarily preserved by its estimators. Using simulated and market data, we show that some traditionally used estimators of $\var$ and $\ES$ do not satisfy CRE properties with non-negligible probability, even if the sampling is from an elliptical distribution.}
\rev{We also show that, even when the underlying risk measure is fixed, different CRE weighting schemes can lead to substantially different estimation accuracy, bias, and capital adequacy.}

For completeness, we discuss next the link between this paper and other results from the risk representation and risk estimation literature, and provide more insight into the underlying regulatory and supervisory background.

\rev{(A)} \rev{There} exists a vast literature on risk measure representation, developed from the ground up using an axiomatic approach, originally introduced in \cite{ArtDelEbeHea1997} and later extended to various setups; cf. \cite{DraKup2013} for an overview of one step risk measures, and \cite{Bielecki2024} for \rev{a} dynamic setup.  Within the axiomatic approach, risk measures can be described in several equivalent ways, allowing both the construction of specific measures and the development of numerical approximation schemes for them. 
As far as we know, this theory, and consequent robust representations for specific families of risk measures, have not been transferred to risk estimation \rev{before, and the interaction between normative properties of risk measures and their estimators has not been investigated,} which is the core topic of this paper.

\rev{(B)} As far as the estimation of $\rho(X)$ for a fixed $\rho$ and/or $X$ is considered, the general goal is to find a formula that is preferably simple and provides a good estimate of the true, unknown value of $\rho(X)$ based on a statistical sample of size $n$, say $\hat\rho_n(X)$. A traditional method to build $\hat\rho_n(X)$, is to approximate the distribution of $X$, or the function $\rho$, or a combination of both. A natural approach is to treat $\hat{\rho}_n(X)$ as a statistical estimate of $\rho(X)$ and to study its properties as the sample size grows. This corresponds to asymptotic properties \rev{induced} by the central limit theorem or estimation consistency analysis, i.e. property $\hat\rho_n(X)\to\rho(X)$, as $n \to \infty$. 
This approach traces back to \cite{Ace2002}, and we refer to the monograph \cite{McNeilFreEmb2015} for a comprehensive literature review; see \cite{BarTan2023} and \cite{BartlEckstein2024} for a more recent comprehensive discussions on contemporary methodologies. Here, we mention that for some classes of estimators that are also discussed in the present work, such as the empirical distribution plug-in
estimators (see Section~\ref{sec:risk_estimators} for precise definition), it was proved that they are consistent and satisfy a central limit theorem type convergence with usual rate $n^{1/2}$, cf.  \cite{Belomestny2012}, and some earlier works, \rev{see} \cite{Weber2007}, \cite{Che2008}, and \cite{Beutner2010}, but for some larger classes of risk measures, the convergence rate is not necessarily $n^{1/2}$, see \cite{BarTan2023}. 
\rev{In these works, financial properties of estimators are typically discussed as a byproduct rather than a design principle.}

\rev{(C)} In the regulatory Internal Model Approach (IMA) for Pillar I bank models, the 10-day $\var$ and 10-day $\ES$ at the confidence levels $\alpha=1\%$ and $\alpha=2.5\%$, respectively, are the key reference market risk capital metrics, see \cite{Bas2006,Bas2019} for details.\footnote{In most practical applications the confidence level $\alpha$ is set to $1\%$, $2.5\%$, or $5\%$. Two main conventions are used when quoting confidence levels for $\var$ and $\ES$: the left-tail convention, adopted in this work and common in the risk measurement literature, and the right-tail convention, which reports thresholds of the form $1-\alpha$ (e.g., $99\%$, $97.5\%$, $95\%$) and is more widespread in risk management and regulatory practice, where the right tail represents losses.} Although the risk measures themselves are fixed, regulatory and supervisory bodies rarely prescribe explicit estimation formulas. Two notable exceptions are the following locally implemented formulas used for Pillar I capital calculations: the minimal Stressed VaR formula under the Basel II framework \cite[Article 10.2]{PRA2020} and the EU Stress Scenario Risk Measure under the Basel III framework \cite[Article 11]{EU2024_397}. More generally, regulatory and supervisory texts specify desired properties of estimators -- such as conceptual soundness, proven backtesting track record, or distribution-free character -- rather than their explicit form, see e.g. \cite[Section 5.3]{EGIM}, \cite[Section 10]{PRA2020}, \cite[Market Risk Standards \& Risk Management Standards]{CBUAE}, or \cite[Section 4.5]{HKMA2024}. In practice, market risk estimators are often non-parametric and based on order statistics from historical simulation P\&Ls, i.e., sorted P\&L sample values. For $\var$ and $\ES$, the estimators are typically linear combinations of order statistics, with coefficients independent of the distribution of $X$; see \cite{EBA2025} for typical VaR look-back period and weighting choices. 

The rest of the paper is organized as follows. Section~\ref{sec:preliminaries} introduces \rev{basic} notation and recalls the definition \rev{and robust representation of coherent risk measures}. \rev{Section~\ref{sec:risk_estimators} introduces coherent risk estimators, discusses their fundamental properties, and provides illustrative examples.} Section~\ref{sec:robust} \rev{presents the main theoretical contributions: the robust representations of CREs (Theorems~\ref{th:coherent_estimator_representation},~\ref{th:law_inv_est_representation}, and~\ref{th:comonotonic_CRE}).} \rev{Section~\ref{sec:consistencyCRE} studies consistency of CREs, with emphasis on spectral risk measures.} \rev{Section~\ref{sec:ES-Lestimator} illustrates the practical relevance of the framework through numerical studies based on simulated and market data. Section~\ref{S:concluding.remarks} collects concluding remarks.}

\section{Preliminaries}\label{sec:preliminaries}

Let $(\Omega,\sG,\bP)$ be a probability space and denote by $L^0:=L^0(\Omega,\sG,\bP)$ the corresponding space of random variables. Throughout, all equalities and inequalities will be understood in $\bP$-a.s. sense. Assume that $\mathcal{X}\subset L^0$ is a vector subspace that contains all constant random variables. Denote by $\mathcal{M}^f:=\mathcal{M}^f(\Omega,\sG)$ the set of finitely additive set functions $Q\colon \sG\to [0,1]$, which are normalized to $1$, i.e. $Q(\Omega)=1$.  We also use the notation $\lfloor b \rfloor:=\max\{k\in \mathbb{Z}\colon k\leq b\}$, for $b\in \mathbb{R}$ and set $\mathcal{M}_n:=\left\{a\in \bR^n: \sum_{i=1}^{n}a_i=1, a_i\geq 0\right\}$, for $n\in \bN$.

For a fixed $n\in\bN$, we use boldface lowercase letters to denote vectors in $\mathbb{R}^n$, e.g. $\bfx=(x_1,\ldots,x_n)\in\bR^n$, and $\langle \bfx, \bfy\rangle:=\sum_{i=1}^nx_iy_i$ stands for the usual dot product in $\bR^n$.  
Moreover, we denote by $s(\mathbf{x}):=(x_{1:n},\ldots, x_{n:n})$ the ordered version of $\mathbf{x}\in\bR^n$, where $x_{i:n}$ is the $i$th smallest element of $(x_1,\ldots, x_n)$, $i=1,\ldots, n$.  
We denote by $S_n$ the set of all  permutations of  $\{1,\ldots, n\}$, and with slight abuse of notation,  we set $\sigma(\mathbf{x}):=(x_{\sigma(i)}, \ldots, x_{\sigma(n)})$, for  $\mathbf{x}\in \mathbb{R}^n$ and $\sigma\in S_n$.

A  risk measuring mapping $\rho\colon \mathcal{X}\to\bR\cup \set{+\infty}$ is called  a \textit{coherent risk measure} (CRM) if it satisfies the following properties:
\begin{enumerate}[(R1)]
\item {\it Monotonicity}, for any $X,Y\in \mathcal{X}$ such that $X\geq Y$, we have $\rho(X)\leq\rho(Y)$;
\item {\it Cash additivity},  for any $X\in \mathcal{X}$ and $m\in\bR$, we have $\rho(X+m)=\rho(X)-m$;
\item {\it Positive homogenity},  for any $X\in \mathcal{X}$ and $\lambda\geq 0$, we have $\rho(\lambda X)=\lambda\rho(X)$;
\item {\it Subadditivity}, for any $X,Y\in \mathcal{X}$, we have $\rho(X+Y)\leq\rho(X)+\rho(Y)$.
\end{enumerate}
Many natural risk measures are law-invariant, in the sense that the value of $\rho(X)$ depends only on the cumulative distribution function (CDF) of $X$ that we denote by $F_X(x):=\bP(X\leq x), x\in \bR$. Formally, $\rho\colon \mathcal{X}\to\bR\cup \set{+\infty}$ is a {\it law-invariant} risk measure if for any $X,Y \in \mathcal{X}$ with the same distribution under $\mathbb{P}$, we have that $\rho(X)=\rho(Y)$.
For law-invariant risk measures, with a slight abuse of notation, we identify $\rho(X)$ with $\rho(F_X)$.

The class of CRMs has been well studied; cf. \cite{FolSch2016} for a comprehensive review in the bounded case $\cX=L^\infty(\Omega,\sG,\bP)$, for both static and dynamic setups, and \cite{DraKup2013} as well as \cite{BieCiaDraKar2013} for a general space.  The postulated properties (R1)–(R4) are both clear and desirable from a risk management perspective. Here, the arguments $X \in \mathcal{X}$ are interpreted as the profit and loss (P\&L) of a financial entity, with $X > 0$ indicating a profit and $X \leq 0$ a loss. Accordingly: (R1) implies that a dominating P\&L entails lower risk; (R2), equivalently expressed as $\rho(X - \rho(X)) = 0$, indicates that $\rho(X)$ is the minimal deterministic capital reserve that neutralizes the risk; (R3) states that risk scales proportionally with the size of the position; and (R4) indicates that diversification reduces risk.  

Among fundamental results in the theory of risk measures are the robust or numerical representations, which allow to express risk measures as suprema over a set of probability measures, thereby linking risk evaluation to the worst-case outcomes for the so-called generalized scenarios; see~\cite[Section 4]{ArtDelEbeHea1999} for an economic view. For convenience, we formulate the result for the bounded and coherent case.

\begin{theorem}[Robust representation of CRMs]\label{th:coherent_measure_representation}
    Let $\cX=L^\infty(\Omega,\sG,\bP)$. A functional \( \rho : \cX \to \mathbb{R} \) is a coherent risk measure if and only if there exists  \( M_\rho \subset \mathcal{M}^{f} \) such that
\begin{equation}
    \rho(X) = \sup_{Q \in M_\rho} \mathbb{E}_Q[-X], \quad X \in \mathcal{X}. 
\end{equation}
Moreover, \(M_\rho \) can be chosen as a convex set for which the supremum is attained. That is, for any $X\in\cX$, there exists $Q_X^*\in M_\rho$ such that $\rho(X) = \bE_{Q_X^*}[-X]$.
\end{theorem}
The proof of Theorem~\ref{th:coherent_measure_representation} can be found, for example, in \cite{FolSch2016}; see also \cite{Del2002} and~\cite{Kus2001}, where the law-invariant case was considered. This theorem has been extended to more general spaces, and larger classes of risk measures; we refer to \cite{DraKup2013} for a comprehensive survey. In particular, one may take $\cX = L^1(\Omega,\sG,\bP)$, where $L^1(\Omega,\sG,\bP)$ is the space of random variables with finite expectation, which encompasses all distributions considered in this paper, including normal or Student’s \( t \)-distributions.

Among the most used and studied risk measures are the value at risk, the expected shortfall, and their weighted generalizations. For completeness, let us now briefly recall selected families of risk measures.

The value at risk ($\var$) at significance level $\alpha\in(0,1)$ is defined as
\begin{equation}\label{eq;var}
    \var_\alpha(X):= \inf \set{m\in \mathbb{R} \mid \bP(X+m<0)\leq \alpha}, \quad X\in\cX. 
\end{equation}
In other words, the $\var$ is the negative of the lower $\alpha$-quantile of $X$, i.e., the right generalized inverse of the cumulative distribution function at $\alpha\in (0,1)$. While widely used, $\var$ is known not to be a CRM due to its lack of subadditivity property (R4). That being said, for linear combinations of risk factors following elliptical distributions, and for confidence levels $\alpha<0.5$, $\var$ is subadditive and thus coherent, see \cite[Theorem~8.28]{McNeilFreEmb2015}.

The expected shortfall (ES) at significance level $\alpha\in(0,1)$ is defined as 
\begin{equation}\label{eq:ES}
\textrm{ES}_\alpha(X):= \frac{1}{\alpha} \int_0^\alpha \var_t(X)\dif t, \quad X\in\cX.
\end{equation}
It is usually interpreted as an average of $\var$s beyond a specific threshold or negative of expected loss beyond $\alpha$-quantile. The second interpretation is motivated by the fact that for continuous random variables $\textrm{ES}_\alpha$ is equal to
\begin{equation}\label{eq:ES-var}
    \textrm{ES}_\alpha(X) = \bE[-X \mid X \leq  - \var_\alpha(X)], \quad X\in\cX,
\end{equation}
see Lemma 2.13 in \cite{McNeilFreEmb2015}. 

Note that the definition of $\textrm{ES}$ could differ in literature and some authors use other names such as {\it average value at risk} or {\it conditional value at risk} to denote the mapping given by either \eqref{eq:ES} or \eqref{eq:ES-var}. Notably, $\ES$ could be seen as a main building block of law-invariant and comonotonic CRMs. Namely, for a fixed probability measure $\mu$ on $[0,1]$, one can define the weighted value at risk (WVaR) given by 
\begin{equation}\label{eq:WVAR}
\wvar_{\mu}(X):=\int_{(0,1]}\ES_\alpha(X)\mu(\dif \alpha), \quad X\in\cX,
\end{equation}
and, typically, any law-invariant and comonotonic CRM could be represented using \eqref{eq:WVAR}; we refer to \cite[Theorem 2.10]{Che2006} for more details. Similarly, one can show that comonotonic and law-invariant risk measures could be constructed directly from $\var$ using the so-called {\it risk spectrum} via the class of {\it spectral risk measures}, see~\cite{Ace2002} for details.

The closed-form formula for $\var$ and $\ES$ is known for many distribution families used in risk management. In particular, by direct computations, one can show that $\var$ and $\ES$ at level $\alpha\in (0,1)$ of a Gaussian random variable $X$ with mean $\mu$ and variance $\sigma^2$, is given by  
\begin{equation}\label{eq:es-normal}
     \var_{\alpha}(X) = -\mu -\sigma \Phi^{-1}(\alpha) \quad\textrm{and}\quad \textrm{ES}_\alpha(X) = -\mu +\sigma \frac{\phi(\Phi^{-1}(\alpha))}{\alpha},
\end{equation}
where $\phi(x):= \frac{1}{\sqrt{2\pi}}\exp(-x^2/2)$, $x\in\bR$, is the density of a standard normal, and $\Phi^{-1}$ is the standard normal quantile, see \cite{McNeilFreEmb2015} for details and more examples.

\section{Coherent risk estimators}\label{sec:risk_estimators}

From a practical standpoint, it is of paramount importance to design a reliable approximation of $\rho(X)$, for a given risk measure $\rho$, using a random sample $\bfx$ of size $n$ from $X$. Similarly to the notion of estimators from statistical analysis, for some given sample size $n\geq 1$, a \textit{risk estimator} is a measurable map $\hat\rho_n\colon \bR^n\to \bR$. Rather than focusing solely on traditional properties from statistical inference (such as consistency or asymptotic normality), this article argues that a good risk estimator should, above all, satisfy properties grounded in sound financial principles. In this section we present some general properties for a coherent risk estimator  $\hat\rho_n$ and a fixed sample size $n\in\bN$, without any reference to the specific choice of $\rho$. 

Similarly to the properties (R1)-(R4) imposed on coherent risk measures, we argue that an estimator of such measures must satisfy similar financially meaningful  properties.  

\begin{definition}[Coherent risk estimator]\label{def:risk.estimator}
A function $\hat\rho_n\colon \bR^n\to \bR$ is a {\it coherent risk estimator} (CRE) if it satisfies 
\begin{enumerate}[(E1)]
\item {\it Monotonicity}, for any $\bfx,\bfx'\in\bR^n$ such that\footnote{For vector order  $\bfx\leq \bfy$ we use the component wise comparison $x_i\leq y_i$, for  $i=1,2,\ldots,n$.}   $\bfx\geq\bfx'$, we have $\hat\rho_n(\bfx)\leq\hat\rho_n(\bfx')$;
\item {\it Cash additivity}, for any $\bfx\in\bR^n$ and $m\in\bR$, we have $\hat\rho_n(\bfx+m)=\hat\rho_n(\bfx)-m$;
\item {\it Positive homogeneity}, for any $\bfx\in\bR^n$ and $\lambda\geq 0$, we have  $\hat\rho_n(\lambda\bfx)=\lambda\hat\rho_n(\bfx)$;
\item {\it Subadditivity}, for any $\bfx,\bfx'\in\bR^n$, we have $\hat\rho_n(\bfx+\bfx')\leq\hat\rho_n(\bfx)+\hat\rho_n(\bfx')$.
\end{enumerate}
\end{definition}
\noindent Coherent risk estimators inherit all axiomatic properties of the coherent risk measures, including their financial meaning. Properties (E1)-(E4) are generic and should hold for any sample points in $\bR^n$. A generic CRE mapping $\hat{\rho}_n$ is a priori not linked or generated by a pre-specified CRM $\rho$, so that one should not expect that $\hat\rho_n$ will converge to any specfic CRM as sample size increases, $n\to\infty$, unless additional conditions are imposed on $\hat{\rho}_n$; see Section~\ref{sec:consistencyCRE} for details. From a practical and regulatory view point, a risk estimator may be interpreted as a mapping that determines the appropriate capital reserve as a function of available data, in contrast to being solely a function that somehow approximates the unknown theoretical value of risk measure.  As we show below, our definition of CRE naturally leads to the important class of non-parametric estimators of baseline risk measures based on order statistics \rev{used e.g. in the formulas for Pillar I capital calculations, see~\cite[Article 42]{EU2024_1085} and \cite[p. 267]{EGIM}.}  

Next, to capture the law-invariant property, we introduce the notion of a law-invariant estimator. We recall that $s(\bfx), \ \bfx\in\bR^n$, denotes the  sorted sample in ascending order.

\begin{definition}[Law-invariant estimator]\label{def:law_invariant_estimator}
A function  $\hat\rho_n\colon \bR^n\to \bR$ is {\it permutation or law-invariant} if, for any $\bfx\in\bR^n$, we have $\hat\rho_n(\bfx)=\hat\rho_n(s(\bfx))$.
\end{definition}
\noindent Note that in Definition~\ref{def:law_invariant_estimator} we may equivalently require that 
$\hat\rho_n(\bfx)=\hat\rho_n(\sigma(\bfx))$ for any $\bfx\in\bR^n$ and permutation $\sigma\in S_n$.  Indeed, directly from the law-invariance property, we get
\[
\hat\rho_n(\sigma(\bfx)) = \hat\rho_n(s(\sigma(\bfx)))=\hat\rho_n(s(\bfx)) = \hat\rho_n(\bfx). 
\]
It should be emphasized that while properties (E1)–(E4) directly mirror the corresponding CRM properties, the relation between the law-invariance of CRMs and CREs is more intricate. In particular, when we assume that the order of sampling can be altered without affecting the estimator, we implicitly induce sampling independence. This is a substantially stronger condition than merely imposing law-invariance of the underlying risk measure. For example, while an estimator is typically law-invariant within i.i.d. sampling framework, this need not hold when the data is generated by a time-dependent process such as {\it Generalized Auto-Regressive Conditional Heteroskedasticity} (GARCH) process, or under an {\it Exponentially Weighted Moving Average} (EWMA) framework, even if the underlying CRM is law-invariant, see e.g. \cite{HanLun2005} and \cite{Car2009} for details. For clarity, throughout most of this paper, we restrict attention to the i.i.d. sampling, though some results--including the core representation result in Theorem~\ref{th:coherent_estimator_representation}--are formulated for the general case.

\rev{Now, we present an example of an ES estimator and show directly that it is a law-invariant CRE.}

\begin{example}[\rev{Sample conditional mean} ES estimator is coherent]\label{ex:ES_coherent}
Let us consider a commonly used non-parametric estimator of the ES at level $\alpha\in (0,1)$, \rev{defined via a sample conditional mean},  given by 
 \begin{equation}\label{eq:ES-param}
 \widehat{\textrm{ES}}^{\rev{\textrm{scm}}}_{\alpha,n}(\bfx):=-\frac{1}{\lfloor n\alpha \rfloor}\sum_{i=1}^{\lfloor n\alpha \rfloor} x_{i:n},
 \end{equation}
where $\bfx\in \mathbb{R}^n$; see \cite{McNeilFreEmb2015}. For simplicity, we assume that $n$ is large enough to have $\lfloor n\alpha\rfloor \geq 1$. \rev{As indicated by the name}, this estimator can be obtained using \eqref{eq:ES-var} and considering the sample conditional mean. One can show it is a coherent and law-invariant risk estimator; see also \cite[Appendix A]{AceTas2002}.  For the sake of completeness, we provide a direct proof that  $\widehat{\textrm{ES}}^{\rev{\textrm{scm}}}_{\alpha,n}$ is a CRE, focusing only on the subadditivity (E4) as the remaining properties are trivially satisfied. We start by introducing the modified indicator function 
 \[
 \1^*_{\{x\leq x_{\lfloor n\alpha \rfloor:n}\}}:= \1_{\{x< x_{\lfloor n\alpha \rfloor:n}\}}+ \1_{\{x =  x_{\lfloor n\alpha \rfloor:n}\}} \frac{\lfloor n\alpha \rfloor-\#\{i\in \{1,\ldots, n\}\colon x_i < x_{\lfloor n\alpha \rfloor:n}\}}{\#\{i\in \{1,\ldots, n\}\colon x_i = x_{\lfloor n\alpha \rfloor:n}\}}
 \]
 for any $x\in \mathbb{R}$ and $\bfx\in \mathbb{R}^n$, which accounts for possible ties in the data. 
Clearly, 
 \[
 \widehat{\textrm{ES}}^{\rev{\textrm{scm}}}_{\alpha,n}(\bfx) = -\frac{1}{\lfloor n\alpha \rfloor}\sum_{i=1}^{n} x_{i} \1^*_{\{x_i\leq x_{\lfloor n\alpha \rfloor:n}\}}.
 \]
 Then, for any $\bfy\in \mathbb{R}^n$ and $\boldsymbol{z}:=\bfx+\bfy$, we obtain
 \begin{equation}\label{eq:ex:ES_coherent:1}
     \lfloor n\alpha \rfloor \left(\widehat{\textrm{ES}}^{\rev{\textrm{scm}}}_{\alpha,n}(\bfx)+\widehat{\ES}^{\rev{\textrm{scm}}}_{\alpha,n}(\bfy)-\widehat{\ES}^{\rev{\textrm{scm}}}_{\alpha,n}(\boldsymbol{z}) \right) = \sum_{i=1}^n x_i\left( \1^*_{\{z_i\leq z_{\lfloor n\alpha \rfloor:n}\}}-\1^*_{\{x_i\leq x_{\lfloor n\alpha \rfloor:n}\}}\right)+\sum_{i=1}^n y_i\left( \1^*_{\{z_i\leq z_{\lfloor n\alpha \rfloor:n}\}}-\1^*_{\{y_i\leq y_{\lfloor n\alpha \rfloor:n}\}}\right).
 \end{equation}
Note that for $i\in \{1,\ldots, n\}$ such that $x_i< x_{\lfloor n\alpha \rfloor:n}$ we have $\1^*_{\{z_i\leq z_{\lfloor n\alpha \rfloor:n}\}}-\1^*_{\{x_i\leq x_{\lfloor n\alpha \rfloor:n}\}}\leq \1^*_{\{z_i\leq z_{\lfloor n\alpha \rfloor:n}\}}-1\leq 0$. Also, for $i\in \{1,\ldots, n\}$ such that $x_i> x_{\lfloor n\alpha \rfloor:n}$ we have $\1^*_{\{z_i\leq z_{\lfloor n\alpha \rfloor:n}\}}-\1^*_{\{x_i\leq x_{\lfloor n\alpha \rfloor:n}\}}\geq \1^*_{\{z_i\leq z_{\lfloor n\alpha \rfloor:n}\}}\geq 0$. Consequently, we deduce
 \[
 \sum_{i=1}^n (x_i-x_{\lfloor n\alpha \rfloor:n})\left( \1^*_{\{z_i\leq z_{\lfloor n\alpha \rfloor:n}\}}-\1^*_{\{x_i\leq x_{\lfloor n\alpha \rfloor:n}\}}\right)\geq 0.
 \]
 Using this inequality and repeating the same argument for $\bfy$, from~\eqref{eq:ex:ES_coherent:1}, we get
\begin{align*}
    \lfloor n\alpha \rfloor \left(\widehat{\ES}^{\rev{\textrm{scm}}}_{\alpha,n}(\bfx)+\widehat{\ES}^{\rev{\textrm{scm}}}_{\alpha,n}(\bfy)-\widehat{\ES}^{\rev{\textrm{scm}}}_\alpha(\boldsymbol{z}) \right) &\geq \sum_{i=1}^n x_{\lfloor n\alpha \rfloor:n}\left( \1^*_{\{z_i\leq z_{\lfloor n\alpha \rfloor:n}\}}-\1^*_{\{x_i\leq x_{\lfloor n\alpha \rfloor:n}\}}\right)+\sum_{i=1}^n y_{\lfloor n\alpha \rfloor:n}\left( \1^*_{\{z_i\leq z_{\lfloor n\alpha \rfloor:n}\}}-\1^*_{\{y_i\leq y_{\lfloor n\alpha \rfloor:n}\}}\right) \\
    &= x_{\lfloor n\alpha \rfloor:n} (\lfloor n\alpha \rfloor-\lfloor n\alpha \rfloor)+y_{\lfloor n\alpha \rfloor:n} (\lfloor n\alpha \rfloor-\lfloor n\alpha \rfloor) =0,
\end{align*}
 where we used the fact that $\sum_{i=1}^n \1^*_{\{x_i\leq x_{\lfloor n\alpha \rfloor:n}\}} = \lfloor n\alpha \rfloor$. This shows that $\widehat{\ES}^{\rev{\textrm{scm}}}_{\alpha,n}(\boldsymbol{x+y})\leq \widehat{\ES}^{\rev{\textrm{scm}}}_{\alpha,n}(\bfx)+\widehat{\ES}^{\rev{\textrm{scm}}}_{\alpha,n}(\bfy)$. \hfill $\square$
\end{example}

\rev{A natural way to construct a law-invariant CRE from a given law-invariant CRM is through a plug-in procedure based on the empirical CDF. Specifically, we define the \textit{empirical risk estimator} $\hat\rho^{\textnormal{emp}}_n:\bR^n\to\bR$ by
\begin{equation}\label{eq:empirical.plugin.estimator}
\hat\rho^{\textnormal{emp}}_n(\bfx):=\rho(\hat{F}^{\textnormal{emp}}_{\bfx}),
\end{equation}
where $\bfx=(x_1,\ldots, x_n)\in \mathbb{R}^n$ and  $\hat{F}^{\textnormal{emp}}_{\bfx}(t):=\frac{1}{n}\sum_{i=1}^n \1_{\{x_i\leq t\}}$, for $t\in \mathbb{R}$. This construction  establishes a direct connection between the risk measure and its empirical analogue, allowing structural properties to carry over to the estimator. In particular, as shown next, when the underlying risk measure is coherent and law-invariant, these properties are preserved under the empirical plug-in procedure, yielding a canonical and structurally consistent coherent risk estimator.}

\begin{theorem}[Empirical risk estimator for CRM is coherent]\label{th:plug_in_ECDF}
    Let $\rho$ be a law-invariant CRM. Then, for any $n\in\bN$, the \rev{empirical risk estimator defined in \eqref{eq:empirical.plugin.estimator}} is a  law-invariant CRE. 
\end{theorem}
\begin{proof} Let $\rho$ be a law-invariant CRM. The law-invariance of $\hat\rho^{\textnormal{emp}}_n$ follows directly from the fact that $\hat{F}^{\rev{\textnormal{emp}}}_{\bfx}=\hat{F}^{\rev{\textnormal{emp}}}_{s(\bfx)}$. Second, we check the cash-additivity (E2), while omitting the  remaining properties that follow by similar arguments.  Let $\bfx\in\bR^n$ and $m\in \mathbb{R}$.  Consider a random variable $Y$ which is uniformly distributed on the set induced by sample $\bfx$, that is, on $\{x_1, \ldots, x_n\}$. Then, noting that $\hat{F}_{\bfx+m}^{\rev{\textnormal{emp}}}$ is the empirical CDF for the random variable $Y+m$, and using the cash additivity of $\rho$, we get
    \[
    \hat\rho^{\textnormal{emp}}_n(\bfx+m)=\rho(\hat{F}_{\bfx+m}^{\rev{\textnormal{emp}}}) = \rho(Y+m)=\rho(Y)-m = \rho(\hat{F}^{\rev{\textnormal{emp}}}_{\bfx})-m = \hat\rho^{\textnormal{emp}}_n(\bfx)-m,
    \]
    which completes the argument. \hfill 
\end{proof}

\rev{Next, we discuss the empirical risk estimators for $\var$ and $\ES$. Recall that $\var$ is not a CRM, as it fails to satisfy subadditivity (R4); this lack of subadditivity carries over to the corresponding empirical plug-in $\var$ estimator.}

\begin{example}[Empirical $\var$ estimator is not coherent]\label{ex:var}
\rev{The empirical risk estimator for $\var$ at level $\alpha\in (0,1)$ is a non-parametric estimator} given by the empirical quantile 
\begin{equation}\label{eq:VAR_estim}
        \widehat\var^{\rev{\textrm{emp}}}_{\alpha,n}(\mathbf{x}):=-x_{(\lfloor \alpha n \rfloor+1):n}, \quad \mathbf{x}\in \mathbb{R}^{n}.
\end{equation}
To illustrate that this estimator is non-coherent we use exemplary parameter values; the example can be easily modified to cover the general case. Namely, let us fix $\alpha=1\%$, $n=100$, and consider  $\mathbf{x}:=(-100,0,\ldots, 0)\in \mathbb{R}^{100}$ and $\mathbf{x}':=(0,-100,0,
    \ldots, 0)\in \mathbb{R}^{100}$. Then, 
\[
    100=\widehat\var^{\rev{\textrm{emp}}}_{1\%,100}(\mathbf{x}+\mathbf{x}')>\widehat\var^{\rev{\textrm{emp}}}_{1\%,100}(\mathbf{x})+\widehat\var^{\rev{\textrm{emp}}}_{1\%,100}(\mathbf{x}')=0+0=0,
\]
and thus the subadditivity property (E4) is violated. Hence,   $ \widehat\var^{\rev{\textrm{emp}}}_{1\%,100}$ is not coherent. 

\rev{More importantly, and somewhat surprisingly, we later show (see Section~\ref{S:ex1}) that the subadditivity property (E4) may be violated with non-negligible positive probability even when the sample is drawn from a multivariate normal distribution. This contrasts with the fact that the theoretical $\var$ satisfies subadditivity (R4) within the class of elliptical distributions, see ~\cite[Theorem~8.28]{McNeilFreEmb2015} for details. Also, this demonstrates that properties of theoretical risk measures do not automatically carry over to their empirical counterparts.}
\hfill $\square$
\end{example}

\begin{example}[Empirical ES estimator is coherent]\label{ex:ES_FRTB}
The empirical risk estimator for ES at level $\alpha\in (0,1)$ can be obtained by replacing $\var_{t}$ by the empirical VaR estimator given in~\eqref{eq:VAR_estim}, for $t\in (0,\alpha)$. After direct integration over $t$ in \eqref{eq:ES}, we obtain the empirical ES estimator
\begin{equation}\label{eq:ES_est_FRTB}
    \widehat{\textrm{ES}}_{\alpha,n}^{\rev{\textrm{emp}}}(\bfx):=-\frac{1}{n\alpha }\left(\sum_{i=1}^{\lfloor n\alpha \rfloor} x_{i:n}+(n\alpha-\lfloor n\alpha \rfloor)x_{(\lfloor n\alpha\rfloor+1):n}\right);
\end{equation}
see also Equation 25 in~\cite{RocUry2002} with $p_k=k/n$  or Article 11 in \cite{EU2024_397}. Alternative ES estimators based on other types of non-parametric quantiles could be also obtained \rev{using a plug-in procedure for the resulting CDF estimator}, see \cite{HynFan1996} and \rev{Section~\ref{S:ex3}}. Finally, we note that while the $\var$ estimator stated in~\eqref{eq:VAR_estim} is not coherent, the corresponding $\ES$ estimator given in~\eqref{eq:ES_est_FRTB} is \rev{a law-invariant CRE by Theorem~\ref{th:plug_in_ECDF}; this will also follow directly from Theorem~\ref{th:coherent_estimator_representation} below.}\qed
\end{example}

Another popular approach for constructing risk estimators is based on the plug-in procedure \rev{used in the parametric setup}: instead of using empirical CDF, we can use a parametric CDF for a pre-specified family of distributions. However, as we show in the following example,  these risk estimators may not be coherent even though the underlying risk measure is coherent. \rev{In fact, as we show in the next section, parametric risk estimators are structurally not coherent; see Theorem~\ref{th:coherent_estimator_representation} for details.}

\begin{example}[
Gaussian ES estimator is not coherent]\label{ex:parametric.nonCRE}
In view of~\eqref{eq:es-normal}, \rev{we can define the parametric plug-in estimator in the Gaussian setup}. Namely, the Gaussian estimator of the ES at level \rev{$\alpha\in (0,1)$ is} defined as
\[
\widehat{\textrm{ES}}^{\textrm{norm}}_{\alpha,n}(\bfx):=-\left(\hat\mu(\bfx)-\hat\sigma(\bfx)\frac{\phi(\Phi^{-1}(\alpha))}{\alpha}\right),
\]
where $\hat\mu(\bfx)$ and $\hat\sigma(\bfx)$ are the sample mean and the sample standard deviation of $\bfx\in \mathbb{R}^n$. As the name suggests, we replaced the true parametric mean and standard deviation in \eqref{eq:es-normal} by their sample estimators. \rev{Clearly, $\widehat{\textrm{ES}}^{\textrm{norm}}_{\alpha,n}$ satisfies (E2), (E3), and property (E4) can be verified by means of the Cauchy-Schwartz inequality. However, monotonicity property (E1) fails to hold true. Indeed,}
\rev{fix $\alpha=1\%$, $n=2$,}  and consider two data samples  $\bfx:=(1,0)$ and $\bfx':=(0,0)$. Clearly, $\bfx \geq \bfx'$,  but
\[
\widehat{\textrm{ES}}^{\textrm{norm}}_{1\%,2}(\bfx)\approx-\left(\frac{1}{2}-\frac{2.66}{\sqrt{2}}\right)\approx 1.38 > 0 = \widehat{\textrm{ES}}^{\textrm{norm}}_{1\%,2}(\bfx').
\]
Thus, $\widehat{\textrm{ES}}^{\textrm{norm}}_{1\%,2}$ is not monotone, and hence not coherent. \rev{More generally, one can show that the monotonicity property is violated with positive probability for Gaussian samples. To illustrate this, we fix $n\in\bN$, $\alpha\in (0,1)$, and consider an i.i.d. sample from a multivariate Gaussian distribution, denoted by $(\bfX,\bfY)=((X_i,Y_i))_{i=1}^{n}$, with the same marginal distributions, that is $X_i,Y_i\sim N(\mu,\sigma^2)$, and correlation $\rho\in (-1,1)$. Noting that $(X_i,Y_i)$ is symmetric, the conditional distribution of $(X_i,Y_i)$ given $X_i\leq Y_i$ is equal to the distribution of $(U_i,V_i) := (\min(X_i, Y_i),\max(X_i,Y_i))$, for $i\in \{1,\ldots,n\}$, see \cite[Section 2]{Vaart1998Book} for details. Thus, setting $(\bfU,\bfV):= ((U_i,V_i))_{i=1}^{n}$, we get 
\[
\bP\left[\widehat{\textrm{ES}}^{\textrm{norm}}_{\alpha,n}(\bfX) > \widehat{\textrm{ES}}^{\textrm{norm}}_{\alpha,n}(\bfY) \mid \bfX \leq \bfY\right] = \bP[\Delta \geq 0],
\]
where $\Delta := \widehat{\textrm{ES}}^{\textrm{norm}}_{\alpha,n}(\bfU) - \widehat{\textrm{ES}}^{\textrm{norm}}_{\alpha,n}(\bfV)$. Next, by direct calculation, noting that $\bfU \sim (2\mu-\bfV)$ due to symmetry, we get $\bE[\hat\sigma(\bfU)]=\bE[\hat\sigma(\bfV)]$, and consequently
\[
\bE[\Delta] = \bE[-\hat\mu(\bfU)+\hat\mu(\bfV)]=\frac{1}{n}\sum_{i=1}^{n}(\bE[V_i]-\bE[U_i])=\bE[V_1]-\bE[U_1] =\left(\mu+\sigma\sqrt{\tfrac{1-\rho}{\pi}}\right)-\left(\mu-\sigma\sqrt{\tfrac{1-\rho}{\pi}}\right)=2\sigma \sqrt{\tfrac{1-\rho}{\pi}}> 0;
\]
see \cite{NadKot2008} for details. This implies $\bP[\Delta > 0]>0$, that is, the monotonicity property is violated for the Gaussian data with positive probability.  Similar calculations and conclusions can be carried over to samples drawn from elliptical distributions.}

\end{example}

\section{Robust representations of a CRE}\label{sec:robust}
In this section, we derive new representations of the CREs, in the spirit of \cite{Del2002} and \cite{Kus2001}, cf. Theorem~\ref{th:coherent_measure_representation}. As already mentioned in Section~\ref{sec:preliminaries}, such representations for risk measures are known as robust or numerical representations, are often linked to dual biconjugates, and are obtained, e.g., via the Fenchel-Moreau theorem, see \cite{DraKup2013}. 

Let us now comment on the significance of these results.  In the statistical setup, they facilitate a full characterization of CREs, and, as we show below, these representations are closely related to the well-studied concept of $L$-estimators. Second, with such results at hand, we can establish additional structural properties of CREs. Third, these representations provide a practical tool for constructing new risk estimators or modifying existing ones. In particular, they enable the design of estimators that satisfy additional desired properties. To the best of our knowledge, the results presented in this section are new. In particular, we are not aware of any systematic studies of estimators defined as suprema over a family of $L$-estimators, a class that plays a central role in our framework.

We start with the generic representation result in which no additional assumptions are imposed on CRE.

\begin{theorem}[Robust representation of CREs]\label{th:coherent_estimator_representation}
    A function $\hat\rho_n\colon \bR^n\to \bR$ is a CRE if and only if there exists a set $M_{\hat\rho_n}^*\subset \mathcal{M}_n$ such that
\begin{equation}\label{eq:coherent_estimator_representation}
    \hat\rho_n(\bfx)= \sup_{a\in M_{\hat{\rho}_n}^*}\langle a,-\bfx\rangle, \quad \mathbf{x}\in \mathbb{R}^n.
\end{equation}
Moreover, \(M_{\hat\rho_n}^* \) can be chosen to be a convex set, independent of $\bfx$, such that the supremum is attained, i.e. for any $ \mathbf{x}\in \mathbb{R}^n$ there exists  $a^*=a^*(\mathbf{x})\in M_{\hat\rho_n}^*$ such that $\hat\rho_n(\mathbf{x})=\langle a^*,-\mathbf{x}\rangle$.  
\end{theorem}

\begin{proof} 
Using properties of the supremum and Definition~\ref{def:risk.estimator}, it is straightforward to check that the map defined in~\eqref{eq:coherent_estimator_representation} is a CRE. Next, we show that  any CRE admits the representation~\eqref{eq:coherent_estimator_representation}. To illustrate this result from different perspectives, we provide two arguments for this part: (a) based on generic properties of convex functionals; (b) based on a suitable identification of risk measures.

\textit{Approach (a)}. 
Combining the positive homogeneity (E3) and the subadditivity (E4)  from Definition~\ref{def:risk.estimator}, we deduce that  $\hat\rho_n$ is convex on $\bR^n$, and in view of \cite[Section 3.2.3]{BoyVan2004}, there exist sets $M_{\hat\rho_n}^*\subset \mathbb{R}^n$ and $B_{\hat\rho_n}^*\subset \mathbb{R}$ such that
\begin{equation}\label{eq:affine_representation}
    \hat\rho_n(\mathbf{x})=\sup_{\substack{a\in M_{\hat\rho_n}^*\\b\in B_{\hat\rho_n}^* }}\left( \langle a,-\mathbf{x}\rangle+b\right), \quad \mathbf{x}\in \mathbb{R}^n,
\end{equation}
and, for any $\mathbf{x}$, the above supremum is attained. By the positive homogeneity  (E3)  with $\lambda=0$, we obtain $0 = \hat\rho_n(0) = \sup_{b\in B_{\hat\rho_n}^*} b$. This  implies that $B_{\hat\rho_n}^*\subset(-\infty,0]$ and there exists a sequence $(b_j)_{ j=1}^\infty$, such that $b_j\in B_{\hat\rho_n}^*$ and $\lim_{j\to\infty}b_j=0$. Consequently, since~\eqref{eq:affine_representation} is given in terms of suprema, we can assume $B_{\hat\rho_n}^*=\{0\}$. Next, for any $i=1,\ldots, n$, let $e_i$ denote the $i$th canonical unit vector in $\mathbb{R}^n$. Then, using the monotonicity (E1), we deduce $0\geq \hat\rho_n(e_i) = \sup_{a\in  M_{\hat\rho_n}^*} \langle a,-e_i\rangle$, for all $i=1,\ldots,n$, and hence, for any $a\in M_{\hat\rho_n}^*$, we have $a\geq 0$. Let us denote by $\mathbf{1}$ the $n$-dimensional vector of ones. Then, by the cash additivity (E2), we obtain
\[
    -1=\hat\rho_n(0)-1=\hat\rho_n(\mathbf{1}) = \sup_{a\in  M_{\hat\rho_n}^*} \langle a,- \mathbf{1}\rangle,
\]
and thus, for any $a=(a_1, \ldots, a_n)\in M_{\hat\rho_n}^*$, we have $\sum_{i=1}^n a_i \geq  1$. On the other hand, we have
\[
    1=\hat\rho_n(0)+1=\hat\rho_n(-\mathbf{1}) = \sup_{a\in  M_{\hat\rho_n}^*} \langle a, \mathbf{1}\rangle,
\]
which implies that $\sum_{i=1}^na_i\leq 1$. Consequently, we get $\sum_{i=1}^n a_i=1$ and conclude the proof.

\textit{Approach (b)}. Let $\hat\Omega:=\{\omega_1,\ldots,\omega_n\}$ be a generic $n$-tuple, and let   $\hat\sG$ be the  family of all subsets of $\hat\Omega$. For  $X\in L^0(\hat\Omega, \hat{\sG})$ we define $\rho(X):=\hat\rho_n((X(\omega_1),\ldots, X(\omega_n)))$. Clearly, $\rho$ satisfies properties (R1)-(R4), as $\hat\rho$ satisfies properties (E1)-(E4), and thus $\rho$ it is a CRM on $(\hat\Omega, \hat{\sG})$. In view of Theorem~\ref{th:coherent_measure_representation},
there exists a family $M_\rho\subset \mathcal{M}^f(\hat\Omega, \hat{\mathcal{G}})$ such that 
    \[
    \hat\rho_n((X(\omega_1),\ldots, X(\omega_n)))=\rho(X)=\sup_{Q\in M_\rho}\mathbb{E}_Q[-X],
    \]
and, for any $X$ the supremum is attained. Since $\hat\Omega$ is finite, any $Q\in M_\rho$ is a probability measure which could be identified with a vector $a:=(Q(\{\omega_1\}), \ldots, Q(\{\omega_n\}))$. Noting that $\mathbb{E}_Q[-X]=\langle a, -(X(\omega_1), \ldots, X(\omega_n))\rangle$, we get~\eqref{eq:coherent_estimator_representation}.  

Finally, by Theorem~\ref{th:coherent_measure_representation}, the convexity $M^*_{\hat\rho_n}$ and the existence of the maximizer $a^*$ follow at once. The proof is complete. 
\end{proof}

\rev{Theorem~\ref{th:coherent_estimator_representation} implies that any CRE must be a supremum over linear combinations of the sample. Hence, structurally different estimators, such as parametric plug-in estimators or GPD tail fits from Section~\ref{S:ex2}, are not expected to be coherent.}

In contrast to Theorem~\ref{th:coherent_measure_representation}, Theorem~\ref{th:coherent_estimator_representation} does not require any additional assumptions on the domain of the underlying mapping. The reason is that for any fixed $n \in \bN$, the realized samples $\bfx$ are elements of the finite-dimensional space $\bR^n$. Consequently, we do not impose any restrictions on the sampling scheme, such as the distribution from which the samples are drawn. 
Also, Theorem~\ref{th:coherent_estimator_representation} accommodates general non-i.i.d. setups, including sampling from time series models or scenario weighting. Nevertheless, in most practical applications, one is typically interested in estimators whose value does not depend on the order of the sampling\rev{, i.e. are law-invariant}.

Let us now derive a version of Theorem~\ref{th:coherent_estimator_representation} for law-invariant CREs. This representation is based on the sorted sample $s(\bfx)$, which reflects an important practical aspect: in real-life risk management applications, the first step is usually to sort the observed P\&Ls (i.e., construct the empirical distribution) before performing the risk computations. We already mentioned that CRE representations are interlinked with $L$-estimators. The next result states that any law-invariant CRE could be represented as \rev{the supremum} over a family of $L$-estimators.

\begin{theorem}[Robust representation of law-invariant CREs]\label{th:law_inv_est_representation} A function $\hat\rho_n\colon \bR^n\to \bR$ is a law-invariant CRE if and only if there exists a set $M^s_{\hat\rho_n}\subset \mathcal{M}_n$ satisfying $a_1\geq a_2\geq \ldots \geq a_n$ for any $a=(a_1, \ldots, a_n)\in M_{\hat\rho_n}^s$, and such that
\begin{equation}\label{eq:law_inv_estimator_representation}
    \hat\rho_n(\bfx)= \sup_{a\in M_{\hat\rho_n}^s}\langle a,-s(\bfx)\rangle, \quad \mathbf{x}\in \mathbb{R}^n.
\end{equation}
Moreover, \(M_{\hat\rho_n}^s \) can be chosen as a convex set for which the supremum is attained, that is, for any $ \mathbf{x}\in \mathbb{R}^n$ there exist weights $a^{s}=a^{s}(\mathbf{x})\in M_{\hat\rho_n}^s$, such that $\hat\rho_n(\mathbf{x})=\langle a^s,-s(\bfx)\rangle$. 
\end{theorem}
\begin{proof}
First, we show that a coherent law-invariant risk estimator $\hat{\rho}_n$ admits the representation~\eqref{eq:law_inv_estimator_representation}. By Theorem~\ref{th:coherent_estimator_representation}, 
there exists a convex set $M^s_{\hat\rho_n}\subset \mathcal{M}_n$ such that $\hat\rho_n(\bfx)= \sup_{a\in M^s_{\hat\rho_n}}\langle a, - \bfx\rangle$, $\forall \mathbf{x}\in \mathbb{R}^n$, and the supremum is attained. Now, we show that the coordinates of $a\in M^s_{\hat\rho_n}$ must be non-increasing. Indeed, by the law-invariance of $\hat\rho_n$, 
\begin{equation}\label{eq:th:law_inv_est_representation:1}
       \hat\rho_n(\bfx)= \hat\rho_n(s(\bfx))=\sup_{a\in M^s_{\hat\rho_n}}\langle a,-s(\bfx)\rangle = \hat\rho_n(\sigma(\bfx)) = \sup_{a\in M^s_{\hat\rho_n}}\langle a,-\sigma(\bfx)\rangle,      
       \quad \mathbf{x}\in \mathbb{R}^n, \ \sigma\in S_n.
\end{equation}
Moreover, we can assume that $M^s_{\hat\rho_n}$ consists only of the elements for which the supremum in~\eqref{eq:coherent_estimator_representation} is attained.  
Hence, for any $a^s\in M^s_{\hat\rho_n}$ we can find $ \mathbf{x}\in \mathbb{R}^n$ such that $\sup_{a\in M^s_{\hat\rho_n}}\langle a,-s(\bfx)\rangle = \langle a^s,-s(\bfx)\rangle$. Then, by \eqref{eq:th:law_inv_est_representation:1}, for any $\sigma \in S_n$, we also have
   \[
   \langle -a^s,\sigma(\bfx)\rangle = \langle a^s,-\sigma(\bfx)\rangle\leq \sup_{a\in M^s_{\hat\rho_n}}\langle a,-\sigma(\bfx)\rangle = \langle a^s,-s(\bfx)\rangle=\langle -a^s,s(\bfx)\rangle.
   \]
From here, in view of \cite[Theorem 369]{HarLitPol1988}, we deduce that the coordinates of $-a^s$ and $s(\bfx)$ have 
the same monotonicity. Since the coordinates of $s(\bfx)$ are non-decreasing, same are the coordinates of $-a^s$. This concludes the proof of this part. 

Next, we show that the map defined in~\eqref{eq:law_inv_estimator_representation} for some fixed set  $M^s_{\hat\rho_n}\subset \mathcal{M}_n$ of vectors with non-increasing coordinates is a law-invariant CRE. The law-invariance property follows at once. As far as coherence, properties (E1)-(E4), we show here only the subadditivity property (E4), since the remaining properties  are straightforward to verify. For any $a\in M^s_{\hat\rho_n}$, using the monotonicity of the coordinates of $a$, we claim that there exists $b=(b_1,\ldots, b_n)\in \mathcal{M}_n$ such that, for any $\mathbf{x}\in \mathbb{R}^n$,  we have
\begin{equation}\label{eq:repr_sum_ES}
    \langle a, -s(\mathbf{x})\rangle = -\sum_{i=1}^n a_i x_{i:n} = \sum_{i=1}^n b_i \widehat{\ES}^{\rev{\textrm{scm}}}_{i/n}(\mathbf{x}),
\end{equation}
where, as in \eqref{eq:ES-param}, we have  $\widehat{\ES}^{\rev{\textrm{scm}}}_{i/n}(\mathbf{x})=-\frac{1}{i}\sum_{j=1}^{i}x_{j:n}$. 
Indeed, 
    \begin{align*}
        -\sum_{i=1}^n a_i x_{i:n} = -\sum_{i=1}^n \left(\sum_{j=i}^{n-1}(a_j-a_{j+1})+a_n \right)x_{i:n} = -\sum_{i=1}^{n-1} \left(\sum_{j=i}^{n-1}(a_j-a_{j+1}) \right)x_{i:n}-\sum_{i=1}^n a_n x_{i:n}. 
    \end{align*}
We note that $-\sum_{i=1}^n a_n x_{i:n} = a_n n \widehat{\ES}^{\rev{\textrm{scm}}}_{n/n}(\mathbf{x})$, and by changing the order of summation above, we also get
\begin{align*}
    -\sum_{i=1}^{n-1} \left(\sum_{j=i}^{n-1}(a_j-a_{j+1}) \right)x_{i:n} = -\sum_{i=1}^{n-1} (a_i-a_{i+1}) \sum_{j=1}^i x_{j:n} = \sum_{i=1}^{n-1} (a_i-a_{i+1}) i \widehat{\ES}^{\rev{\textrm{scm}}}_{i/n}(\mathbf{x}).
\end{align*}
Thus, setting $b_i:=(a_i-a_{i+1})i$, $i=1, \ldots, n-1$, and $b_n := na_n$ we obtain~\eqref{eq:repr_sum_ES}. We remark that $ b$ is independent of $\bfx$, and by direct calculation we also have $\sum_{i=1}^n b_i = \sum_{i=1}^n a_i = 1$. Thus,  by the monotonicity of $(a_i)$ we also obtain that  $b_i\geq 0$, so $b=(b_1, \ldots, b_n)\in \mathcal{M}_n$.
    
    By Example~\ref{ex:ES_coherent}, the map $\mathbf{x}\to\widehat{\ES}^{\rev{\textrm{scm}}}_{i/n}(\mathbf{x})$ is a CRE, for any $i$. Consequentially, for any $\mathbf{x}, \bfy\in\mathbb{R}^n$, using~\eqref{eq:repr_sum_ES}, we obtain
    \[
    \langle a, -s(\mathbf{x}+\mathbf{y})\rangle = \sum_{i=1}^n b_i \widehat{\ES}^{\rev{\textrm{scm}}}_{i/n}(\mathbf{x+y})\leq \sum_{i=1}^n b_i \widehat{\ES}^{\rev{\textrm{scm}}}_{i/n}(\mathbf{x})+\sum_{i=1}^n b_i \widehat{\ES}^{\rev{\textrm{scm}}}_{i/n}(\mathbf{y}) = \langle a, -s(\mathbf{x})\rangle+\langle a, -s(\mathbf{y})\rangle.
    \]
Finally, taking here the supremum over $a\in M_{\hat\rho_n}^s$, we deduce that  $\hat\rho_n(\mathbf{x}+\mathbf{y})\leq \hat\rho_n(\mathbf{x})+\hat\rho_n(\mathbf{y})$, which concludes the proof.
\end{proof}

\begin{remark}
The weights set $M^s_{\hat{\rho}_n}$  in the representation \eqref{eq:law_inv_estimator_representation}  is generally not unique. To provide an illustrative example, set $\hat\rho_n(\bfx):= -\min_i x_i$,  for $\mathbf{x}\in \mathbb{R}^n$. This is a law-invariant CRE since $\min_i x_i = x_{1:n}$. Now, let $M':=\{(1,0,\ldots, 0)\}$ and $M'':=\{(1-1/k,1/k,0, \ldots, 0)\colon k\in \mathbb{N}, k\geq 2\}$. Then, 
    $\hat\rho_n(\bfx)= \sup_{a\in M'}\langle a,-s(\bfx)\rangle = \sup_{a\in M''}\langle a,-s(\bfx)\rangle$, for any $\mathbf{x}\in \mathbb{R}^n$.
\end{remark}
The representation result in Theorem~\ref{th:law_inv_est_representation} is consistent with the corresponding result for law-invariant CRMs obtained in \cite{Kus2001}. However, this does not imply that a supremum over $L$-statistics should always be used when estimating a law-invariant CRM, since law-invariance of CREs depends both on the estimation method and on the underlying  CRM itself (cf. the comment following Definition~\ref{def:law_invariant_estimator}).

In the next section, we further examine the connection between $L$-estimators and CREs. In particular, we show that under comonotonicity, the supremum in Theorem~\ref{th:law_inv_est_representation} can be omitted. Before doing so, for completeness, we describe the relationship between the sets $M_{\hat\rho_n}^*$ and $M_{\hat\rho_n}^s$ from Theorems~\ref{th:coherent_estimator_representation} and \ref{th:law_inv_est_representation}, and present some illustrative examples.

\begin{proposition}[Link between robust representations for general and law-invariant CREs] \label{prop:MM}
    Let $\hat\rho_n\colon \bR^n\to \bR$ be a law-invariant CRE admitting representation $\hat\rho_n(\bfx)= \sup_{a\in M^s_{\hat\rho_n}}\langle a,-s(\bfx)\rangle$, where $\mathbf{x}\in \mathbb{R}^n$ and $M^s_{\hat\rho_n}\subset \mathcal{M}_n$ is such that $a_1\geq a_2\geq \ldots \geq a_n$ for $a\in M^s_{\hat\rho_n}$. Then, we have
    \[
    \hat\rho_n(\bfx) = \sup_{a\in M_{\hat\rho_n}^\sigma}\langle a,-\bfx\rangle, \quad \mathbf{x}\in \mathbb{R}^n,
    \]
where $M_{\hat\rho_n}^\sigma:=\{\sigma(a)\colon \sigma \in S_n, a\in M^s_{\hat\rho_n}\}$.
\end{proposition}
\begin{proof}
Note that, for any $\sigma \in S_n$, $a\in \mathcal{M}_n$ and $\mathbf{x}\in \mathbf{R}^n$, we have $\langle \sigma(a),\mathbf{x}\rangle = \langle a,\sigma^{-1}(\mathbf{x})\rangle$, where $\sigma^{-1}$ is the inverse permutation. Then, we obtain
\begin{align*}
    \sup_{a\in M_{\hat\rho_n}^\sigma}\langle a,-\bfx\rangle = \sup_{\substack{a\in M^s_{\hat\rho_n}\\\sigma \in S_n}}\langle \sigma(a),-\bfx\rangle = \sup_{\substack{a\in M^s_{\hat\rho_n}\\\sigma \in S_n}}\langle a,-\sigma^{-1}(\bfx)\rangle\geq \sup_{\substack{a\in M^s_{\hat\rho_n}}}\langle a,-s(\bfx) \rangle= \hat\rho_n(\mathbf{x}), 
\end{align*}
where the inequality follows from the fact that $s(\mathbf{x}) = \sigma_0^{-1}(\bfx)$ for some permutation $\sigma_0\in S_n$. On the other hand, using the rearrangement inequality \cite[Theorem 368]{HarLitPol1988}), for any $a\in M^s_{\hat\rho_n}$ and $\sigma \in S_n$, we obtain
\[
\langle a,-\sigma(\bfx)\rangle\leq \langle a,-s(\bfx)\rangle,
\]
which concludes the proof.
\end{proof}
\rev{Proposition~\ref{prop:MM} relates the two robust representations by showing that the weights for the unsorted sample in Theorem~\ref{th:coherent_estimator_representation} are obtained by permuting the sorted-sample weights from Theorem~\ref{th:law_inv_est_representation}, so that $M_{\hat\rho_n}^\sigma$ provides an explicit construction of the set $M_{\hat\rho_n}^*$.}

Now, we show \rev{examples of} the sets $M_{\hat\rho}$ for some specific families of risk measures and show how they are related to estimation formulas.

\begin{example}[Robust representation of \rev{the sample conditional mean} ES estimator]\label{ex:ES_representation}
 
 Let us fix $n\in\bN$, $\alpha\in (0,1)$, and consider a non-parametric estimator $\widehat{\ES}^{\rev{\textrm{scm}}}_{\alpha,n}$ defined in Example~\ref{ex:ES_coherent}, see \eqref{eq:ES-param}. Then, it is easy to show that $\widehat{\ES}^{\rev{\textrm{scm}}}_{\alpha,n}$ admits a law-invariant robust representation from Theorem~\ref{th:law_inv_est_representation} with the set
\[
    M^{s}_{\widehat{\ES}^{\rev{\textrm{scm}}}_{\alpha,n}}:=\left\{(a_1,\ldots, a_n)\colon a_i:=
        \frac{1}{\lfloor n\alpha \rfloor}\1_{\{i\leq \lfloor n\alpha \rfloor\}},\, i=1,2\ldots,n
    \right\}.
\]
\hfill $\square$
\end{example}

\begin{example}[Robust representation of CREs based on order statistics]\label{ex:ES_inequal_weights}
The weighting scheme introduced in Example~\ref{ex:ES_representation} that leads to estimator $\widehat{\ES}^{\rev{\textrm{scm}}}_{\alpha,n}$ could be modified. In particular, this could lead to alternative ES estimators such as $\widehat{\ES}^{\rev{\textrm{emp}}}_{\alpha,n}$ defined in \eqref{eq:ES_est_FRTB}.  Namely, for a fixed $n\in\bN$ and $\alpha\in (0,1)$, let us consider the risk estimator
 \begin{equation}\label{eq:ES-q_based}
         \widehat{R}^{q}_{\alpha,n}(\bfx):=-\sum_{i=1}^{\lfloor n\alpha \rfloor+1} q_i x_{i:n},
 \end{equation}
where a single $q:=(q_1, \ldots, q_{\lfloor n\alpha \rfloor+1}, 0,\ldots, 0)\in \mathcal{M}_n$ is fixed and such that $q_1\geq q_2\geq \ldots\geq q_{\lfloor n\alpha \rfloor+1}$. Then, from Theorem~\ref{th:law_inv_est_representation} with the supremum being the single element, we get that
\[
M^{s}_{\widehat{R}^{q}_{\alpha,n}}:=\{q\}
\] 
is a robust representation set for $\widehat{R}^{q}_{\alpha,n}$, and this measure is a law-invariant CRE. \rev{We refer to Section~\ref{S:ex3} for a practical discussion on the choice of $q$ for ES estimator.}
\qed\end{example}

\begin{example}[Robust representation of CREs based on suprema of order statistics]\label{ex:robust_ES}
The class of law-invariant CREs considered in Example~\ref{ex:ES_inequal_weights} could be further generalized by considering the suprema of weighted order statistics. Let us consider the risk estimator 
    \begin{equation}\label{eq:R.sup.3}
        \widehat{R}_{\alpha,n}(\bfx):=\sup_{q\in Q}\widehat{R}^q_{\alpha,n}(\bfx).
    \end{equation}
where $Q\subset \cM_n$ is such that any $q\in Q$ satisfies the same conditions as in Examples~\ref{ex:ES_inequal_weights}, and $\widehat{R}^q_{\alpha,n}$ is defined in \eqref{eq:ES-q_based}. Then, from Theorem~\ref{th:law_inv_est_representation}, we get that 
\[
M^{s}_{\widehat{R}_{\alpha,n}}:=Q
\]
is a robust representation set for $\widehat{R}_{\alpha,n}$, and this risk measure a law-invariant CRE. \hfill $\square$
\end{example}

In contrast to Example~\ref{ex:ES_representation} and Example~\ref{ex:ES_inequal_weights}, the order statistic weighting scheme in Example~\ref{ex:robust_ES} could depend on a sample realization, that is, different values of $q$ could attain a supremum in \eqref{eq:R.sup.3} for different samples $\bfx\in\bR^n$. 

To provide further illustration, let us consider a CRM that has a different structure than ES, \rev{namely the expectile value at risk ($\ExpVAR$)}, and study the dual representation of the corresponding plug-in CRE.

\begin{example}[Robust representation of \rev{the empirical ExpVaR} estimator]\label{ex:exp_var} 

In this example we consider expectile value at risk ($\ExpVAR$) family of risk measures indexed by a significance level $\alpha\in (0,1/2)$. This family identifies an important class of law-invariant risk measures which are both CRM and elicitable, see~\cite{BelBer2017,BelNegPya2019,EmbSchWan2022} for more details. The $\ExpVAR$ at significance level $\alpha\in (0,1/2)$ is given by
    \begin{equation}
        \ExpVAR_\alpha(X):=-\argmin_{c\in \mathbb{R}}\left(\alpha \mathbb{E}[(X-c)_+^2]+(1-\alpha)\mathbb{E}[(X-c)_-^2] \right),\quad X\in\cX,
    \end{equation}
    where $(b)_+:=\max(b,0)$ and $(b)_-:=\max(-b,0)$. For $\cX=L^1$, we have $\ExpVAR_\alpha(X)=-e_\alpha(X)$, where $e_\alpha(X)$ is the $\alpha$-expectile of $X$, that is, a unique solution to the equation  $\alpha \mathbb{E}[(X-e_\alpha(X))_+]-(1-\alpha)\mathbb{E}[(X-e_\alpha(X))_-]=0$. Using this representation and Theorem~\ref{th:plug_in_ECDF}, we can implicitly define an empirical $\ExpVAR$ estimator by setting
\begin{equation}\label{eq:ExpVAR.est}
\widehat\ExpVAR^{\rev{\textrm{emp}}}_{\alpha,n}(\mathbf{x}):=-\hat{e}_\alpha(\mathbf{x}),
\end{equation}
where the empirical expectile $\hat{e}_\alpha(\mathbf{x})$ is defined as a solution to equation
\begin{equation}\label{eq:expectile_emp}
        \alpha \frac{1}{n}\sum_{i=1}^n[x_i-\hat{e}_\alpha(\mathbf{x})]_+-(1-\alpha) \frac{1}{n}\sum_{i=1}^n[x_i-\hat{e}_\alpha(\mathbf{x})]_-=0.
\end{equation}
\rev{Note that the application of $\ExpVAR_\alpha$ to the empirical distribution replaces expectations by empirical averages, so the estimator \eqref{eq:ExpVAR.est} indeed coincides with the empirical estimator defined in \eqref{eq:empirical.plugin.estimator}.}  Now, let $n^*(\mathbf{x})$ be such that 
    \begin{equation}\label{eq:ExpVaR:n_star}
        n^\star(\mathbf{x}):=\sup\{k\in \{1,\ldots,n\}\colon  x_{k:n}\leq \hat{e}_\alpha(\mathbf{x})\}.
    \end{equation}
    Then, we can rewrite~\eqref{eq:expectile_emp} as
\[
    (1-\alpha) \sum_{i=1}^{n^\star(\mathbf{x})}(x_{i:n}-\hat{e}_\alpha(\mathbf{x}))+
        \alpha \sum_{i=n^\star(\mathbf{x})+1}^n(x_{i:n}-\hat{e}_\alpha(\mathbf{x}))=0.
\]
Hence, $\hat{e}_\alpha(\mathbf{x})$ satisfies
\[
\hat{e}_\alpha(\mathbf{x})= \frac{1-\alpha}{(1-2\alpha)n^\star(\mathbf{x})+n\alpha}\sum_{i=1}^{n^\star(\mathbf{x})} x_{i:n}+
        \frac{\alpha }{(1-2\alpha)n^\star(\mathbf{x})+n\alpha}\sum_{i=n^\star(\mathbf{x})+1}^n x_{i:n}.
\]
Consequently, $\widehat\ExpVAR_{\alpha,n}^{\rev{\textrm{emp}}}$ admits the following representation
\begin{equation}
    \widehat\ExpVAR^{\rev{\textrm{emp}}}_{\alpha,n}(\mathbf{x})=-\left(\sum_{i=1}^{n^\star(\mathbf{x})} \frac{1-\alpha}{(1-2\alpha)n^\star(\mathbf{x})+n\alpha}x_{i:n} +\sum_{i=n^\star(\mathbf{x})+1}^n \frac{\alpha}{(1-2\alpha)n^\star(\mathbf{x})+n\alpha}x_{i:n}\right)=\langle a^\star(\mathbf{x}),-s(\mathbf{x})\rangle,
\end{equation}
where  $a^\star_i(\mathbf{x}):=
    \frac{(1-\alpha)}{(1-2\alpha)n^\star(\mathbf{x})+n\alpha}$ for $i=1, \ldots, n^\star(\mathbf{x})$, and $
    a^\star_i(\mathbf{x}):=\frac{\alpha}{(1-2\alpha)n^\star(\mathbf{x})+n\alpha} $ for $ i=n^\star(\mathbf{x})+1, \ldots, n$. As we later illustrate in Example~\ref{ex:exp_var:2}, $a^\star(\mathbf{x})$ is different for different samples $\mathbf{x}$. Using the fact that expectile value at risk is coherent, we can also recover the robust representation of $\widehat\ExpVAR^{\rev{\textrm{emp}}}_{\alpha,n}$ from  Theorem~\ref{th:law_inv_est_representation}. Indeed, one can show that
\begin{equation}
    \widehat\ExpVAR^{\rev{\textrm{emp}}}_{\alpha,n}(\mathbf{x})=\sup_{a\in M^s_{\widehat\ExpVAR^{\rev{\textrm{emp}}}_{\alpha,n}}}\langle a, -s(\mathbf{x})\rangle,  
\end{equation}
    where $M^s_{\widehat\ExpVAR^{\rev{\textrm{emp}}}_{\alpha,n}}:=\{a^\star(\mathbf{x})\colon \mathbf{x}\in \mathbb{R}^n \}$.\hfill $\square$
\end{example}

\subsection*{Comonotonic CREs and their representation as L-estimators} \label{sec:Lest}
We recall that in statistical analysis an \textit{$L$-estimator} is a linear combination of the order statistics $(x_{i:n})_{i=1,\ldots,n}$, see~\cite[Section~22]{Vaart1998Book} or \cite{DavidNagaraja2003} for details. Thus, certain risk estimators -- such as the ES estimator $\widehat{\ES}^{\rev{\textrm{scm}}}_{\alpha,n}(\bfx)$ given in Example~\ref{ex:ES_coherent} or the $\var$ estimator $\widehat\var^{\rev{\textrm{emp}}}_{\alpha,n}(\bfx)$ given in Example~\ref{ex:var} -- are specific instances of $L$-statistics. More generally, in view of Theorem~\ref{th:law_inv_est_representation} and Proposition~\ref{prop:MM}, a law-invariant CRE is a supremum over a set of $L$-estimators. 
For any generic risk measure $\rho$, from both practical and computational perspectives, it is desirable for the set $M^s_{\hat\rho_n}$ in \eqref{eq:law_inv_estimator_representation} to be small -- ideally a singleton -- as is the case in Example~\ref{ex:ES_representation}, where $\widehat{\ES}^{\rev{\textrm{scm}}}_{\alpha, n}(\bfx)$ is represented via the singleton set $M^s_{\widehat{\ES}^{\rev{\textrm{scm}}}_{\alpha,n}}$. Indeed, for any $a\in\cM_n$ such that $a_1\geq a_2\geq\ldots\geq a_n$, the value $\langle a,-s(\bfx)\rangle = -\sum_{i=1}^n a_i x_{i:n}$ has a natural interpretation, since it could be seen as an empirical form of an average (or weighted) VaR, in which VaR estimates are represented by order statistics. This shows that such estimators are related to a large and important class of CRMs; cf.~\cite[Theorem 2.5]{Ace2002},  \cite[Section~4.4]{FolSch2016}, and Remark~\ref{rem:spectral.risk.measures} for more details. The aim of this section is to provide sufficient conditions for $M^s_{\hat\rho_n}$ to be a singleton, using the comonotonicity property.

 \rev{The risk measure $\rho$ is {\it comonotonic} (or {\it comonotonic additive}) if it satisfies the property $\rho(X+Y)=\rho(X)+\rho(Y)$ for comonotonic random variables $X$ and $Y$; recall that $X$ and $Y$ are said to be comonotonic if $(X(\omega)-X(\omega'))(Y(\omega)-Y(\omega')) \ge 0$ for all $\omega,\omega' \in \Omega$. This additive-type property reflects the principle that aggregation of perfectly aligned random  positions should not generate artificial risk reduction or amplification and it leads to a specific type of robust representation for law-invariant risk measures, see~\cite{Kus2001}. }
 
 \rev{To formulate and study this property in the context of risk estimators, we recall that two vectors $\mathbf{x}, \mathbf{y}\in \mathbb{R}^n$ are \textit{comonotonic} if $(x_i-x_j)(y_i-y_j)\geq 0$, for $i,j=1,\ldots, n$.
In other words, the coordinates of $\mathbf{x}$ and $\mathbf{y}$ are jointly increasing or decreasing. Similar to risk measures, comonotonicity extends to the risk estimation framework: the estimator should be additive for vectors with the same order, and perfectly aligned samples, when combined, should not artificially reduce or increase the estimated risk. }

\begin{definition}[Comonotonic estimator]\label{def:comonotonic_estimator}
A function $\hat\rho_n\colon \bR^n\to \bR$ is {\it comonotonic} if $ \hat\rho_n(\mathbf{x}+\mathbf{y}) = \hat\rho_n(\mathbf{x})+\hat\rho_n(\mathbf{y})$ for any comonotonic vectors $\bfx\in\bR^n$ and $\mathbf{y}\in \mathbb{R}^n$.  
\end{definition}
\noindent \rev{Note that if a law-invariant CRM $\rho$ is comonotonic, then its empirical estimator is also comonotonic; if vectors are comonotonic, then the induced random variables under the (empirical) discrete uniform measure are comonotonic, and the additivity property of $\rho$ transfers directly to the empirical estimator.} As we show next, for any comonotonic law-invariant CRE, the set $M^s_{\hat\rho_n}$ can be chosen as a singleton, a result that may be viewed as a version of \cite[Theorem~7]{Kus2001} adapted to CREs. 

\begin{theorem}[Robust  representation of comonotonic and law-invariant  CREs]\label{th:comonotonic_CRE}
    A risk estimator $\hat\rho_n\colon \bR^n\to \bR$ is a comonotonic law-invariant CRE if and only if there exists a unique $a=(a_1, \ldots, a_n)\in\mathcal{M}_n$ satisfying $a_1\geq a_2\geq \ldots \geq a_n$ and
\begin{equation}\label{eq:comonotonic_representation}
    \hat\rho_n(\bfx)= \langle a,-s(\bfx)\rangle, \quad \mathbf{x}\in \mathbb{R}^n.
\end{equation}
\end{theorem}
\begin{proof}
$(\Leftarrow)$ Note that Theorem~\ref{th:coherent_estimator_representation} and Theorem~\ref{th:law_inv_est_representation} imply that $\hat\rho_n$ defined in~\eqref{eq:comonotonic_representation} is a law-invariant CRE. By the definition of comonotonicity, the functions $\bfx\mapsto x_{i:n}$,  with $i=1, \ldots,n$, are comonotonic. Thus, the map $\hat\rho_n$ is comonotonic as the (negative) convex combination of comonotonic functions.

\noindent $(\Rightarrow)$ Assume that $\hat{\rho}_n$ is a comonotonic law-invariant CRE, and let ${M}^s_{\hat\rho_n}$ be any representing set from Theorem~\ref{th:law_inv_est_representation}. We define
    \[
    {N}_{\hat\rho_n}(\mathbf{x}):=\{a\in {M}^s_{\hat\rho_n}\colon \hat\rho_n(\mathbf{x}) = \langle a,-s(\mathbf{x})\rangle\}, \quad \mathbf{x}\in \mathbb{R}^n.
    \]
In view of Theorem~\ref{th:law_inv_est_representation} and the continuity of the map $a\mapsto \langle a,-s(\mathbf{x})\rangle$, for any $\mathbf{x}\in \mathbb{R}^n$, the set ${N}_{\hat\rho_n}(\mathbf{x})$ is non-empty and closed. We show that
    \begin{equation}\label{eq:th:comonotonic_CRE:intersection}
        \bigcap_{\mathbf{x}\in \mathbb{R}^n}{N}_{\hat\rho_n}(\mathbf{x}) \neq \emptyset.
    \end{equation}
Then, any $a\in \bigcap_{\mathbf{x}\in \mathbb{R}^n}{N}_{\hat\rho_n}(\mathbf{x})$ satisfies~\eqref{eq:comonotonic_representation}.

To prove~\eqref{eq:th:comonotonic_CRE:intersection}, it is enough to show that for any $K\in \mathbb{N}$, $K\geq 2$, and $\mathbf{x}_1, \ldots, \mathbf{x}_K\in \mathbb{R}^n$ we have \begin{equation}\label{eq:th:comonotonic_CRE:intersection_finite}
    \bigcap_{i=1}^K {N}_{\hat\rho_n}(\mathbf{x}_i) \neq \emptyset.
    \end{equation}
Indeed, if~\eqref{eq:th:comonotonic_CRE:intersection_finite} holds and $\bigcap_{\mathbf{x}\in \mathbb{R}^n}{N}_{\hat\rho_n}(\mathbf{x}) = \emptyset$, then we have $\bigcup_{\mathbf{x}\in \mathbb{R}^n}(\cM_n\setminus {N}_{\hat\rho_n}(\mathbf{x}) ) = \cM_n$. However, since $\cM_n$ is compact and any $\cM_n\setminus {N}_{\hat\rho_n}(\mathbf{x})$ is open, we may find $\mathbf{x}_1, \ldots, \mathbf{x}_K\in \mathbb{R}^n$ such that $\bigcup_{i=1}^K(\cM_n\setminus {N}_{\hat\rho_n}(\mathbf{x}_i) ) = \cM_n$, which contradicts~\eqref{eq:th:comonotonic_CRE:intersection_finite}. 
Thus, to show~\eqref{eq:comonotonic_representation}, it is enough to show~\eqref{eq:th:comonotonic_CRE:intersection_finite}. Hence, let $K\in \mathbb{N}$, $K\geq 2$, $\mathbf{x}_1, \ldots, \mathbf{x}_K\in \mathbb{R}^n$, and let us define $\mathbf{x}:=\sum_{i=1}^K s(\mathbf{x}_i)$. Also, note that for any $k=1, \ldots, K-1$, the vectors $\sum_{i=1}^k s(\mathbf{x}_i)$ and $s(\mathbf{x}_{k+1})$ are comonotonic. By comonotonicity and law-invariance, we inductively get
\begin{equation}\label{eq:th:comonotonic_CRE:1}
    \textstyle\hat\rho_n(\mathbf{x}) = \sum_{i=1}^K \hat\rho_n(s(\mathbf{x}_i)) = \sum_{i=1}^K \hat\rho_n(\mathbf{x}_i).
\end{equation}
Next, let $a\in {N}_{\hat\rho_n}(\mathbf{x})$, and since $s(\mathbf{x})=\mathbf{x}$, we deduce
\[
   \textstyle \hat\rho_n(\mathbf{x}) = \langle a,-s(\mathbf{x})\rangle = \langle a,-\mathbf{x}\rangle = \sum_{i=1}^K \langle a,-s(\mathbf{x}_i)\rangle.
\]
 Next, note that from ${N}_{\hat\rho_n}(\mathbf{x})\subset {M}^s_{\hat\rho_n}$, we have  $\hat\rho_n(\mathbf{x}_i)\geq \langle a,-s(\mathbf{x}_i)\rangle$. In fact, recalling~\eqref{eq:th:comonotonic_CRE:1}, we obtain $\hat\rho_n(\mathbf{x}_i)= \langle a,-s(\mathbf{x}_i)\rangle$ for any $i=1, \ldots, K$. Thus, $a\in {N}_{\hat\rho_n}(\mathbf{x}_i)$ for any $i=1, \ldots, K$, which concludes the proof of~\eqref{eq:comonotonic_representation}.

Finally, we show that $a$ from~\eqref{eq:comonotonic_representation} is unique. Let $a^1,a^2\in \mathcal{M}_n$ be such that
\[
    \hat\rho_n(\mathbf{x}) = \langle a^1,-s(\bfx)\rangle=\langle a^2,-s(\bfx)\rangle, \quad \mathbf{x}\in \mathbb{R}^n.
\]
    Then, setting $\mathbf{x}:=(-1,0,\ldots, 0)$ we obtain $a_1^1 = \hat\rho_n(x) = a_1^2 $. Next, setting $\mathbf{x}:=(-1,-1,\ldots, 0)$, we get
    $a_1^1+a_2^1 = \hat\rho_n(x) = a_1^2+a_2^2 $, so $a_2^1 = a_2^2$. Thus, we inductively obtain $a_i^1=a_i^2$, $i=1, \ldots, n$, which concludes the proof.
\end{proof}
\rev{Practically, Theorem~\ref{th:comonotonic_CRE} implies that, within the class of law-invariant comonotonic CRMs, the set over which the supremum is taken in the robust representation can be reduced to a singleton, so that Kusuoka’s representation is directly mirrored at the estimation level. Importantly, this also ensures computational tractability, as it reduces the computation of the risk estimator to a single L-estimator.}

We conclude this section with two examples. In the first example, we recall the usual way of estimating spectral risk measures and show that the corresponding risk estimators are comonotonic and law-invariant CREs, while in the second example we present a numerical illustration that one cannot find unique weights for non-comonotonic risk measure estimators.

\begin{example}[\rev{Spectral} CRE for spectral risk measures]\label{ex:spectral_RM} As stated in Section~\ref{sec:preliminaries}, the class of WVaR risk measures could be represented using \textit{spectral risk measures}, see~\cite{Ace2002} for details. A spectral risk measure is given by
\begin{equation}\label{eq:spectral_rm}
\textstyle    \rho(X) = -\int_0^1 \var_{\alpha}(X) \phi(\alpha) d\alpha,
\end{equation}
where the \textit{risk spectrum} $\phi\colon [0,1]\to \mathbb{R}_+$ is (weakly) decreasing, bounded, and $\int_0^1 \phi(t)dt= 1$. In order to estimate \eqref{eq:spectral_rm}, we can consider the discretised version of the risk spectrum. Namely, for any $n>1$, set $a_{i,n}:=\int_{\frac{i-1}{n}}^{\frac{i}{n}} \phi(s)ds$, $i=1,\ldots, n$, and consider the risk estimator given by
\begin{equation}\label{eq:estimator_representation}
\textstyle    \hat\rho^a_n(\mathbf{x}) = -\sum_{i=1}^n a_{i,n} x_{i:n}.
\end{equation}
Clearly, due to the properties of the risk spectrum, we have $a_{i,n} \geq a_{i+1,n}$, for $i=1,\ldots,n-1$, and $\hat{\rho}_n^a$ admits representation \eqref{eq:comonotonic_representation}, so that it is a law-invariant and comonotonic CRE. This CRE could be seen as a natural non-parametric plug-in estimator of the corresponding spectral risk measure \eqref{eq:spectral_rm}, similar to the CRE discussed in Theorem~\ref{th:plug_in_ECDF}. In the next section, we establish further important properties of this estimator.\hfill $\square$
\end{example}

\begin{example}[Non-comonotonicity of \rev{the empirical} ExpVaR estimator]\label{ex:exp_var:2} 
Let $\widehat\ExpVAR_{\alpha,n}^{\rev{\textrm{emp}}}$ be the law-invariant CRE defined in \eqref{eq:ExpVAR.est}. Non-comonotonicity of this estimator follows from Theorem~\ref{th:comonotonic_CRE}  and Example~\ref{ex:exp_var}, where the maximizer has been shown to be dependent on the sample. For completeness, let us now numerically illustrate the non-comonotonicity of $\widehat\ExpVAR_{\alpha,n}^{\rev{\textrm{emp}}}$. Let $\alpha=1/4$, $\mathbf{x}:=(1,2,3)$, and $\mathbf{y}:=(0,0,1)$. Clearly, $\mathbf{x}$ and $\mathbf{y}$ are comonotonic. Also, routine calculations show
\[
\widehat\ExpVAR^{\rev{\textrm{emp}}}_{1/4,3}(\mathbf{x}):=1.6, \quad \widehat\ExpVAR^{\rev{\textrm{emp}}}_{1/4,3}(\mathbf{y})\approx 0.1429, \quad \widehat\ExpVAR^{\rev{\textrm{emp}}}_{1/4,3}(\mathbf{x}+\mathbf{y})=1.8,
\]
which directly shows non-comonotonicity as $\widehat\ExpVAR^{\rev{\textrm{emp}}}_{1/4,3}(\mathbf{x}+\mathbf{y})\neq \widehat\ExpVAR^{\rev{\textrm{emp}}}_{1/4,3}(\mathbf{x})+\widehat\ExpVAR^{\rev{\textrm{emp}}}_{1/4,3}(\mathbf{y})$.    \hfill $\square$
\end{example}

\begin{remark}[Robust representation weights for CRE and risk spectrum]\label{rem:spectral.risk.measures}
In Example~\ref{ex:spectral_RM}, we illustrated the inherent relationship between the risk spectrum in the spectral representation of CRMs and the structure of estimation weights in the robust representation of CREs. Specifically, the vectors $a \in \mathcal{M}_n$ defined in Theorem~\ref{th:comonotonic_CRE} and Theorem~\ref{th:law_inv_est_representation} can be interpreted as approximations of risk spectra: they mimic the weakly decreasing, bounded, unit-integral properties and are applied to order statistics, which approximate empirical quantiles, i.e., $\var$ at different significance levels. That said, the link does not amount to a strict equivalence, since plug-in spectral estimators for CRMs rely on empirical quantile representations, whereas risk spectra may also be estimated via alternative approximation schemes; cf. Example~\ref{ex:spectral_RM} and Example~\ref{es:con.alternative.spectral}. 
\end{remark}

\section{Consistency of CREs for i.i.d. samples}\label{sec:consistencyCRE}

In this section, we focus on the problem how to generate a sequence of CREs $\hat\rho_n(\mathbf{X}_n)$ that approximates a given CRM $\rho(X)$, \rev{where $\mathbf{X}_n:=(X_1,X_2,\ldots,X_n)$, for $n\in\bN$ and $\mathbf{X}:=(X_1,X_2,\ldots,)$ is an i.i.d. sample from the distribution of $X\in \mathcal{X}$.} We start by stating the definition of consistent risk estimators.

\begin{definition}[Consistent estimator]
    A sequence of risk estimators $(\hat\rho_n)_{n=1}^\infty$ is \textit{consistent} for a risk measure $\rho\colon \mathcal{X}\to \mathbb{R}\cup\set{+\infty}$ if, for any $X\in \mathcal{X}$ such that $\rho(X)<+\infty$,  and i.i.d. sample $\bfX$ from the distribution of $X$, we have
    \begin{equation}\label{eq:fin:consistency}
    \hat\rho_n(\bfX_n)\xrightarrow{\rev{\bP}}\rho(X), \quad n\to\infty,
    \end{equation}
where $\xrightarrow{\rev{\bP}}$ stands for the \rev{convergence in probability with respect to $\bP$. A sequence $(\hat\rho_n)_{n=1}^\infty$ is {\it strongly consistent} if the convergence in \eqref{eq:fin:consistency} holds almost surely.}
\end{definition}

\begin{remark}[Consistency under ergodic stationarity]
\rev{For simplicity, throughout this work we assume that $\bfX$ is an i.i.d.\ sequence. However, the notion of consistency of estimators naturally extends to broader classes of random sequences with intertemporal dependence. In particular, we conjecture that the definition of consistency, together with Proposition~\ref{pr:consistent_estimator}, can be established under suitable ergodic stationarity assumptions, see \cite[Section 2]{Hay2011}. This typically encompasses canonical time series models that admit the strong mixing property, such as autoregressive (AR), moving average (MA), or generalized autoregressive conditional heteroskedastic (GARCH) processes. A rigorous treatment of these extensions in the context of risk estimation lies beyond the scope of the present paper and will be pursued in future work.}
\end{remark}

As the next result shows,  consistency of the risk estimators preserves coherence of the limiting risk measure. 

\begin{proposition}[CREs consistent limit leads to CRM]\label{pr:consistent_estimator}
    Suppose that there exists a consistent sequence of CREs $(\hat\rho_n)_{n=1}^\infty$ for $\rho\colon \mathcal{X}\to \mathbb{R}\cup\set{+\infty}$. Then, $\rho$ is a law-invariant CRM.
\end{proposition}
\begin{proof}
    To show that $\rho$ is CRM, we only show the subadditivity condition; the remaining properties are proved similarly or are straightforward \rev{by using monotonicity and the continuous mapping theorem}. Let $(X_i,Y_i)_{i=1}^n$ be an i.i.d. bivariate sample from $(X,Y)\in\cX\times \cX$. Then, $(Z_i)_{i=1}^n$ with $Z_i=X_i+Y_i$ is an i.i.d. sample from $X+Y$ and using the consistency and the coherence of $\hat\rho_n$, we have
    \[
    \rho(X+Y)=\rev{\bP-\!}\lim_{n\to\infty}\hat\rho_n(\bfZ_n)\leq \rev{\bP-\!} \lim_{n\to\infty}\left(\hat\rho_n(\bfX_n)+\hat\rho_n(\bfY_n)\right) = \rho(X)+\rho(Y),
    \]
    where we set $\bfX_n:=(X_1, \ldots, X_n)$, $\bfY_n:=(Y_1, \ldots, Y_n)$, $\bfZ_n:=(Z_1, \ldots, Z_n)$. 
    
    To prove law-invariance, let $X,Y\in \mathcal{X}$ that have the same distribution. Then, an  i.i.d. sample $\bfX_n=(X_i)_{i=1}^n$ from $X$ is also an i.i.d. sample from $Y$, and by consistency, we have
    \[
    \rho(X)=\rev{\bP-\!}\lim_{n\to\infty}\hat\rho_n(\bfX_n) = \rho(Y),
    \]
    which concludes the proof.
\end{proof}

Note that in Proposition~\ref{pr:consistent_estimator} we do not require $\hat\rho_n$ to be law-invariant to get the law-invariance of $\rho$.  Also from Proposition~\ref{pr:consistent_estimator}, we observe that if $\rho$ is not coherent, then there is no consistent sequence of CREs for $\rho$.

Following the discussion in Section~\ref{sec:Lest}, our goal is to find CREs represented as an $L$-estimator that are consistent. As the next result show, there is a strong connection between consistency of $L$-estimators and spectral risk measures discussed in Example~\ref{ex:spectral_RM}. Recall~\eqref{eq:spectral_rm} for the definition of a spectral risk measure with the risk spectrum $\phi$.

\begin{theorem}[Consistency of CREs based on \rev{risk spectrum}]\label{th:consistency}
    Let $\rho$ be a spectral risk measure with the risk spectrum $\phi$. Let  $\hat\rho_n, \ n>1$, be a risk estimator given by 
\begin{equation}\label{eq:estimator_representation2}
\textstyle    \hat\rho_n(\mathbf{x}) = -\sum_{i=1}^n a_{i,n} x_{i:n},
\end{equation}
where $a^n:=(a_{1,n}, \ldots, a_{n,n})\in \mathcal{M}_n$ with $a_{1,n}\geq a_{2,n}\geq \ldots \geq a_{n,n}$. Put $\phi_n(t):=\sum_{i=1}^n na_{i,n} \1_{\{t\in (\frac{i-1}{n},\frac{i}{n}]\}}$, for $n>1$, $t\in [0,1]$, and assume that $\sup_{n}\sup_{t\in [0,1]}\phi_n(t)<\infty$. Then, the following conditions are equivalent:
    \begin{enumerate}
        \item $\hat\rho_n$ is a \rev{strongly} consistent estimator of $\rho$ on $\cX=L^1$.
        \item For any $t\in (0,1)$, we have $\int_0^t \phi_n(s)ds\to \int_0^t \phi(s)ds$, as $n\to\infty$.
    \end{enumerate}
\end{theorem}
\begin{proof}
The claim follows from~\cite{Zwe1980}, Corollary~2.1 and the subsequent discussion, by setting $J_N = \phi_N$, $J = \phi$, and $g = F_X^{-1}$.
\end{proof}

\rev{The estimator's strong consistency assumption in the first condition of  Theorem~\ref{th:consistency} could be relaxed. Namely, if one assumes that the limit in the second condition exists, then one can replace strong consistency with consistency; it is sufficient to note that a sequence convergent in probability has an almost sure convergent subsequence.} Alternative conditions to the uniform boundedness of the sequence $(n a_n)$ from Theorem~\ref{th:consistency} can be found in the extensive literature on the consistency of $L$-estimators; for a comprehensive review, we refer the reader to~\cite{AarBur1996, MiaMa2021}, \cite[Chapter 8]{Ser1980}, and \cite{Mas1982}.

Next we consider an example of  risk estimator that satisfies the assumptions of Theorem~\ref{th:consistency}.

\begin{example}[Consistency of \rev{spectral} CRE]\label{ex:consistenty.srm}
    Let $\rho$ be a spectral risk measure with the risk spectrum $\phi$, and let $\hat{\rho}^a_n$ be its CRE defined in Example~\ref{ex:spectral_RM}. We claim that this estimator is strongly consistent for $\rho$. Indeed, for any $n>1$ and $i=1,\ldots, n$ we have $n a_{i,n} = n\int_{\frac{i-1}{n}}^{\frac{i}{n}} \phi(s)ds \leq n \frac{1}{n}\Vert \phi\Vert_1 = \Vert \phi\Vert_1$.
    Thus, setting $\phi_n(t):=\sum_{i=1}^n na_{i,n} \1_{\{t\in (\frac{i-1}{n},\frac{i}{n}]\}}$, we obtain
    \[
    \sup_{n}\sup_{t\in [0,1]}\phi_n(t) = \sup_{n}\sup_{i=1,\ldots, n} na_{i,n}\leq \Vert \phi\Vert_1<\infty.
    \]
    Also, for any $t\in (0,1)$, we deduce
    \[
    \int_0^t \phi_n(s)ds = n \sum_{i=1}^{[tn]} a_{i,n}\frac{1}{n}+n a_{[tn],n}\left(t-\frac{[tn]}{n}\right) = \int_0^{\frac{[tn]}{n}}\phi(s)ds + (tn-[tn])\int_{\frac{[tn]-1}{n}}^{\frac{[tn]}{n}} \phi(s)ds \to \int_0^t \phi(s)ds, \quad n\to\infty.
    \]
    Then, by Theorem~\ref{th:consistency}, strong consistency of $\hat\rho_n^a$ follows. \hfill $\square$ 
\end{example}

\begin{example}[Consistency of alternative plug-in CREs for spectral risk measures]\label{es:con.alternative.spectral}
The risk spectrum $\phi$ could be approximated using different weighting schemes. 
Let us consider the setup introduced in  Example~\ref{ex:spectral_RM} but with alternative weights defined by as $a_{i,n} := \frac{\phi(i/n)}{\sum_{k=1}^n \phi(k/n)}$, $n>1$, $i=1,\ldots,n$, we refer to~\cite[Section 5]{Ace2002} where this approximation scheme is introduced and discussed. Using a similar argument as in Example~\ref{ex:consistenty.srm} one can show that the corresponding risk estimator is also strongly consistent; see also Theorem~5.4 in \cite{Ace2002}.\hfill $\square$
\end{example}

The consistency of estimators for general risk measures has been well studied in the literature. Broadly speaking, using the language of this manuscript, these results fall into two (overlapping) categories: non-parametric plug-in estimator (e.g. Example~\ref{ex:ES_coherent}) and empirical estimators (e.g. Theorem~\ref{th:plug_in_ECDF}). We emphasize that in all these works, the focus has been on statistical asymptotic properties (such as consistency and rates of convergence) and on certain selected economic or financial properties (such as robustness and elicitability). In contrast, the present work concentrates on comprehensive risk management properties of these estimators. 

For the sake of completeness, we review some of the existing key results.  We recall that a law-invariant CRM $\rho$, under some mild conditions, e.g. from \cite{Kus2001}, admits the representation 
\begin{equation}\label{eq:law-inv-CRM}
    \rho(X) =\sup_{\mu\in\cM}\wvar_{\mu}(X)= \sup_{\mu\in\cM}\int_{(0,1]} \ES_{\alpha}(X) \mu(\dif \alpha),  
\end{equation}
for some set $\cM\subset\cM^f$. Similar to Example~\ref{ex:spectral_RM}, using a natural discrete approximation of the integrals as well as replacing  $\ES_{\alpha}, \alpha\in(0,1]$, by a given family of estimators  $\widehat{\ES}_\alpha,  \alpha\in(0,1]$, we can consider the estimator 
\begin{equation}\label{eq:CRE-law-inv-general}
    \hat{\rho}_n(\bfx) : = \sup_{\mu\in\cM} \sum_{i=1}^n  \widehat{\ES}_{\alpha_i}(\bfx)   \mu ( (\alpha_i,\alpha_{i+1}]),
\end{equation}
where $(\alpha_i)$ forms a uniform partition of $[0,1]$. Equivalently, after some direct algebraic transformations, it can be written as 
\[
\hat{\rho}_n(\bfx) = \sup_{a\in M^*} \langle a,-s(\bfx) \rangle,
\]
for some explicitly computed class of weights $M^*$, i.e. supremum over a class of $L$-estimators. We note that while there is a vast literature on asymptotic properties of $L$-estimators, those methods rarely can be extended directly to the supremum of a set of $L$-estimators.  In \cite[Theorem 3.15]{PflugWozabal2010}, the authors prove, under some fairly general assumptions, that $\hat\rho_n$ given by \eqref{eq:CRE-law-inv-general} is consistent, and asymptotically normal with rate of convergence $n^{1/2}$; see also \cite{Wozabal2009}. In \cite{ConDegSca2010} the authors study the robustness and sensitivity of similar CREs.

For an arbitrary law-invariant CRM $\rho$, in view of Theorem~\ref{th:plug_in_ECDF}, one can build a law-invariant CRE, which we call the empirical (plug-in) risk estimator. In \cite{Belomestny2012} the authors show that these estimators are consistent and satisfy  a central limit theorem with usual rate $n^{1/2}$; the manuscript also considers non-i.i.d. data. We refer to  \cite{Weber2007,Che2008,Beutner2010}, for some earlier works on this topic. Finally, we mention \cite{BarTan2023} that investigates the same class of empirical estimators but for a larger class of law-invariant risk measures, where the authors show that generally speaking, the rate of convergence is not necessarily classical $n^{1/2}$. We also refer to \cite{BarTan2023} for a comprehensive and relatively up to date literature review on this topic.

\section{\rev{Numerical study: subadditivity violations and evaluation of weighting schemes}}\label{sec:ES-Lestimator}

In this section, we provide three numerical illustrations that highlight the practical relevance of the results developed in Section~\ref{sec:risk_estimators} and Section~\ref{sec:robust}. The examples serve two purposes. First, we demonstrate that the desirable normative properties of risk measures -- most notably subadditivity -- may fail at the estimation level even when the underlying population risk measure satisfies them. Second, we show how the structural characterization of CREs translates into materially different capital outcomes through alternative admissible weighting schemes. All examples are conducted in a regulatory-style setup with learning period length $n=250$ and confidence levels consistent with Basel-type applications ($\alpha=1\%$ for VaR and $\alpha=2.5\%$ for ES). In addition to simulated data, we use empirical market data from the Fama--French data library; see~\cite{FamFre2026} for details. Specifically, we employ daily average value-weighted returns for {\it 100 Portfolios Formed on Size and Book-to-Market}. This dataset provides a broad cross-section of equity portfolios with heterogeneous risk characteristics and is widely used in empirical asset pricing and risk analysis. For empirical investigations of subadditivity, we evaluate all pairwise combinations of the 100 portfolios using daily returns from 2025, resulting in 4950 portfolio pairs.

The first example investigates subadditivity violations of the empirical VaR estimator, both under simulated Gaussian samples and on market data. The second example considers a semi-parametric ES estimator based on a Generalised Pareto tail fit and illustrates that non-coherent ES estimators may also violate subadditivity in practice. The third example restricts attention to coherent ES estimators and compares several admissible $L$-estimator weight structures, thereby showing  how different CRE representations affect estimation accuracy, bias, and capital protection. Taken together, these examples demonstrate that coherence of a theoretical risk measure does not automatically carry over to its estimator, and that within the class of coherent estimators, the selection of the weight structure represents a structurally meaningful modeling choice. 

\subsection{Subadditivity violation for empirical VaR estimator -- normal samples and market data}\label{S:ex1}

In this section, we consider the empirical VaR estimator $\widehat\var_{\alpha,n}^{\textrm{emp}}$ defined in Example~\ref{ex:var} and check if this estimator could genuinely violate the subadditivity property on both simulated and market data. In line with the regulatory framework, we fix $n=250$ and $\alpha=1\%$. From a theoretical point of view, noting that the set $D:=\{(\mathbf{x}, \mathbf{x}')\in \mathbb{R}^{2n}\colon \widehat\var_{\alpha,n}^{\textrm{emp}}(\mathbf{x}+\mathbf{x}')>\widehat\var_{\alpha,n}^{\textrm{emp}}(\mathbf{x})+\widehat\var_{\alpha,n}^{\textrm{emp}}(\mathbf{x}')\}$ is non-empty and open, due to the continuity of $\mathbf{x}\mapsto \widehat\var_{\alpha,n}^{\textrm{emp}}(\mathbf{x})$, we immediately deduce that the subadditivity condition is violated with positive probability for any i.i.d. sample from a full-support distribution. Our goal is to show that the violation probability is statistically material, and could occur even in the Gaussian setup under which the theoretical VaR is a subadditive risk measure.

First, we consider simulated samples $(\mathbf{x}, \mathbf{x}')$ from a bivariate Gaussian distribution with zero means, unit variances, and various correlation parameters $\rho\in \{0, 0.1, \ldots, 0.9\}$. We perform a Monte Carlo study for $N=10\,000$ simulations (separately for each $\rho$), each of size $n=250$, and evaluate the subadditivity ratio statistic given by
\begin{equation}\label{eq:VaR_ratio}
    r_{\var}(\mathbf{x}, \mathbf{x}'):=\frac{\widehat\var^{\mathrm{emp}}_{\alpha,n}(\mathbf{x}+\mathbf{x}')}{\widehat\var^{\mathrm{emp}}_{\alpha,n}(\mathbf{x})+\widehat\var^{\mathrm{emp}}_{\alpha,n}(\mathbf{x}')}-1;
\end{equation}
note that subadditivity violation occurs if and only if this ratio is (strictly) positive. The MC-induced violation probabilities are presented in Figure~\ref{fig:ex:1}, left panel.
\begin{figure}[htp!]
    \centering
    \includegraphics[width=0.33\linewidth]{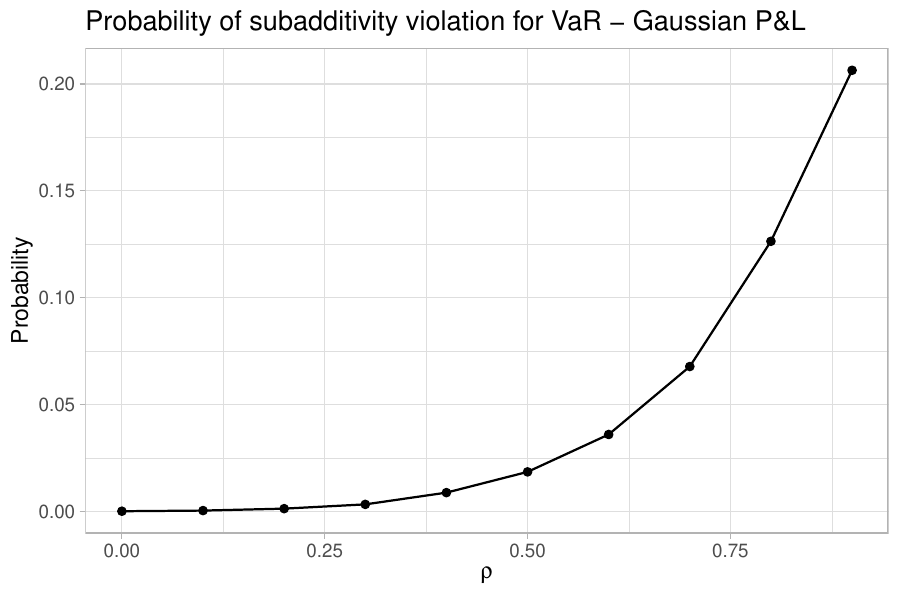}
    \includegraphics[width=0.33\linewidth]{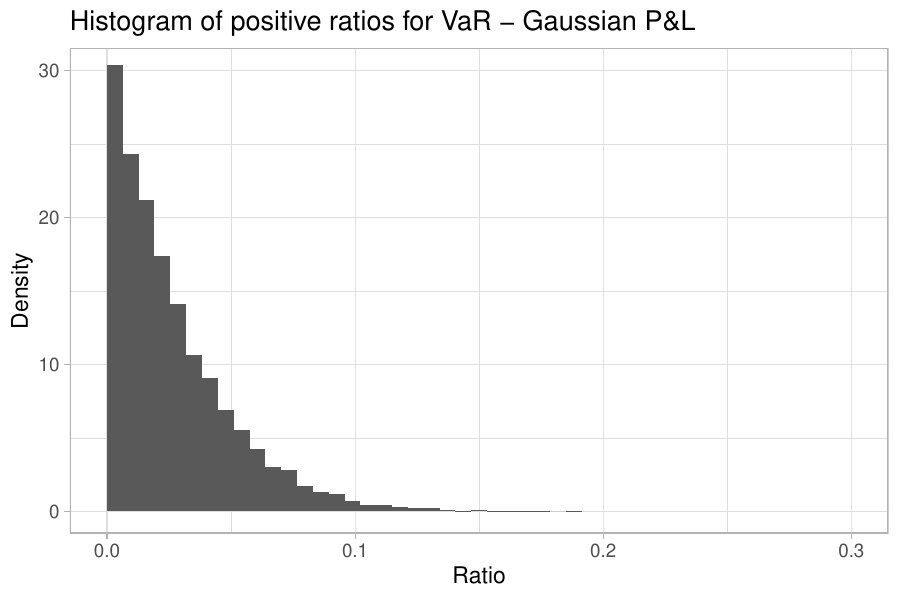}
    \includegraphics[width=0.33\linewidth]{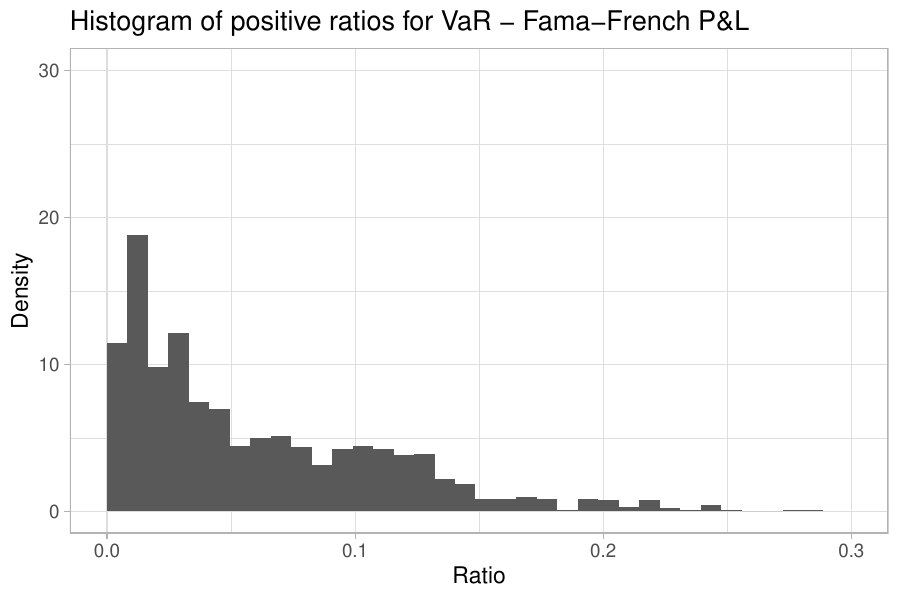}
    \caption{Left panel:  simulated probability of the subadditivity violation for $\var$ estimator discussed in Section~\ref{S:ex1} with respect to the correlation coefficient of the underlying bivariate normal distribution. Middle panel: the histogram of positive ratio values for the statistic defined in~\eqref{eq:VaR_ratio}, for the normally distributed data with $\rho=0.9$. Right panel: an analogous histogram for the Fama-French dataset.}
    \label{fig:ex:1}
\end{figure}
In particular, we note that the violation probability for
$\rho=0.9$ (a plausible correlation for equity portfolio market data)  is larger than 20.5\%, which implies that even in the bivariate Gaussian setup, the probability of subadditivity violation could be substantial. This observation should be contrasted with the fact that the (theoretical) VaR risk measure is subadditive for elliptical distributions. We note that our observation aligns with~\cite{DanJorSam2013}, which analyzed the discrepancy between theoretical subadditivity and its violation in simulation studies.

Second, to further study the  relative impact of the violation, we investigated the histogram of the violation ratio statistic for $\rho=0.9$; see Figure~\ref{fig:ex:1}, middle panel. We note that  the size of the violation could be substantial, often larger than 5\%, and in extreme cases even bigger than 10\%. This implies that a lack of subadditivity can not only appear in the sample data but also significantly affect the risk analysis or distort the risk assessment, for example, when evaluating diversification benefits.

Finally, we checked the subadditivity violation for the empirical $\var$ estimator using the Fama--French dataset. For each pair of the 100 portfolios in the dataset, and one-day returns from 2025, we computed the ratio defined in~\eqref{eq:VaR_ratio}. Subadditivity was violated in 22.4\% of the total 4950 pairs. The histogram of the ratio statistic, shown in the right panel of Figure~\ref{fig:ex:1}, indicates that the relative impact of the violation is much more severe than in theoretical Gaussian data, with over 20\% of the sample exceeding the 10\% threshold.  Also, the average correlation between all pairs of portfolios was 73\%, confirming the observation that high correlation could lead to a more pronounced subadditivity violation. These show that the empirical $\var$ estimator can substantially violate the subadditivity normative property on market data, and the extent of the violation can be considerable.

It should be emphasized that, while on the theoretical level, the violation of $\var$ subadditivity is often presented on rather synthetic discrete random variables and does not occur for elliptical distributions, it is much more frequent and pronounced in sample estimation, even when a realistic sample size ($n=250$) and confidence level (1\%) is considered. This highlights the need for caution when attempting to directly transfer theoretical risk measure properties to empirical risk estimators.

\subsection{Subadditivity violation for a parametric ES estimator based on GPD -- Student's $t$ samples and market data}\label{S:ex2}

In this section, we want to assess whether the subadditivity property could be violated for  an exemplary non-coherent risk estimator of ES showing the practical relevance of  Theorem~\ref{th:comonotonic_CRE}. To this end, we consider a semi-parametric ES estimator based on the Generalised Pareto Distribution (GPD) fitted to the left tail of the underlying distribution. For a regulatory-style setup with $n=250$ and $\alpha=2.5\%$, we consider the GPD plug-in ES estimator given by
\begin{equation}\label{eq:gpd.es}
\widehat\ES^{\mathrm{GPD}}_{\alpha,n}(\mathbf{x}):= \frac{\widehat\var^{\mathrm{GPD}}_{\alpha,n}(\mathbf{x})}{1-\hat\xi} + \frac{\hat\beta + \hat\xi u}{1 - \hat\xi},
\end{equation}
where  $\widehat\var^{\mathrm{GPD}}_{\alpha,n}(\mathbf{x}):= -u + \frac{\hat\beta}{\hat\xi} \left( \left(\frac{\alpha}{p_u}\right)^{-\hat\xi} - 1 \right)$ is the GPD plug-in VaR estimator, $u$ is the threshold parameter corresponding to the empirical 20\%  sample quantile,\footnote{The key conclusions remain the same for other levels in the reasonable range.} $p_u:=\frac{1}{n}\sum_{i=1}^n 1_{\{ x_i\leq u\}}$ is the threshold sample ratio count, and $(\hat\beta,\hat\xi)\in \bR_+\times (-\infty,1)$ are GPD distribution parameters fitted to the exceedance vector (i.e. positive coordinates of $u-\mathbf{x}$). The parameters were fitted using the standard maximum log-likelihood procedure; we refer to Section 5.2 in~\cite{McNeilFreEmb2015} for more details on this estimator and its links to Extreme Value Theory. As we later show on both simulated and market data, the GPD plug-in estimator defined in~\eqref{eq:gpd.es} is not a CRE. 

\begin{figure}[htp!]
    \centering
    \includegraphics[width=0.33\linewidth]{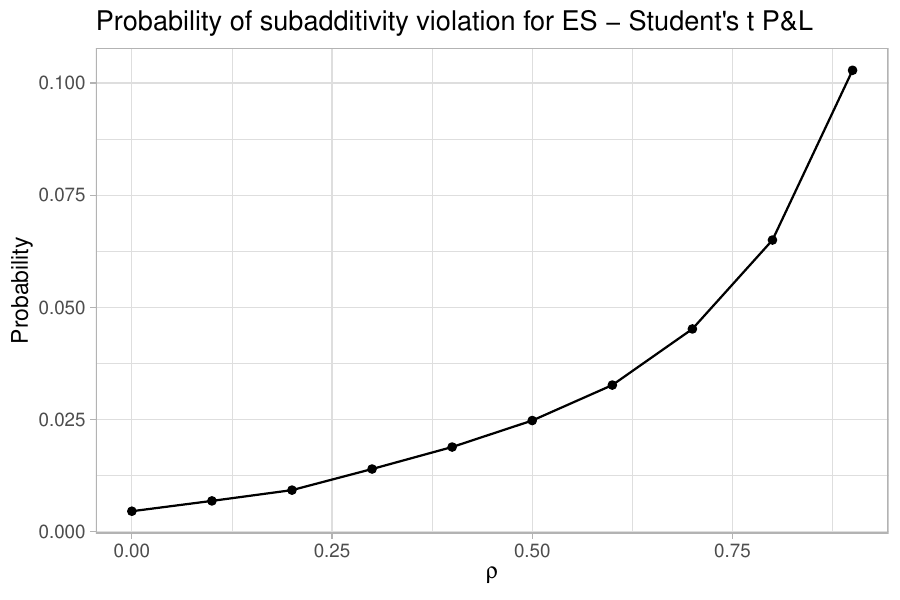}
    \includegraphics[width=0.33\linewidth]{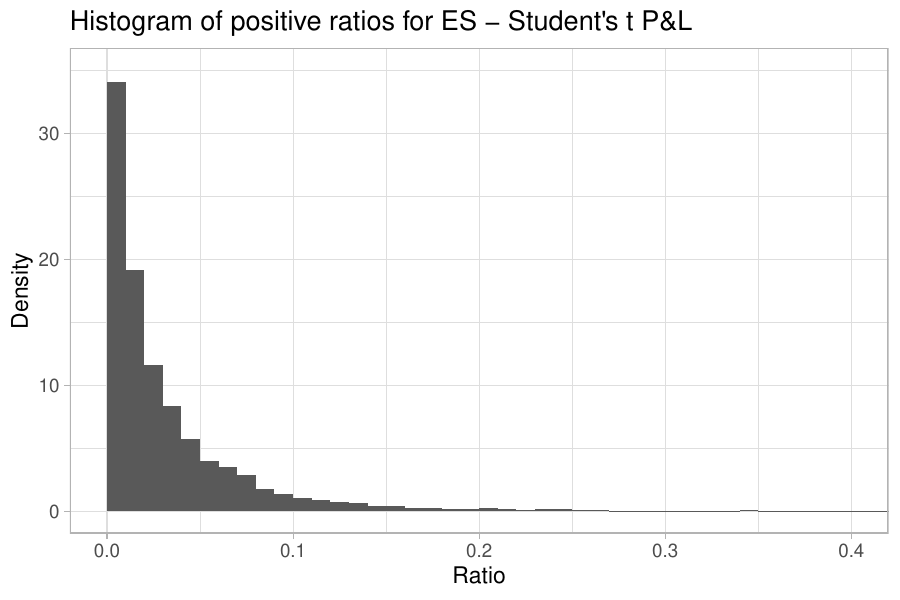}
    \includegraphics[width=0.33\linewidth]{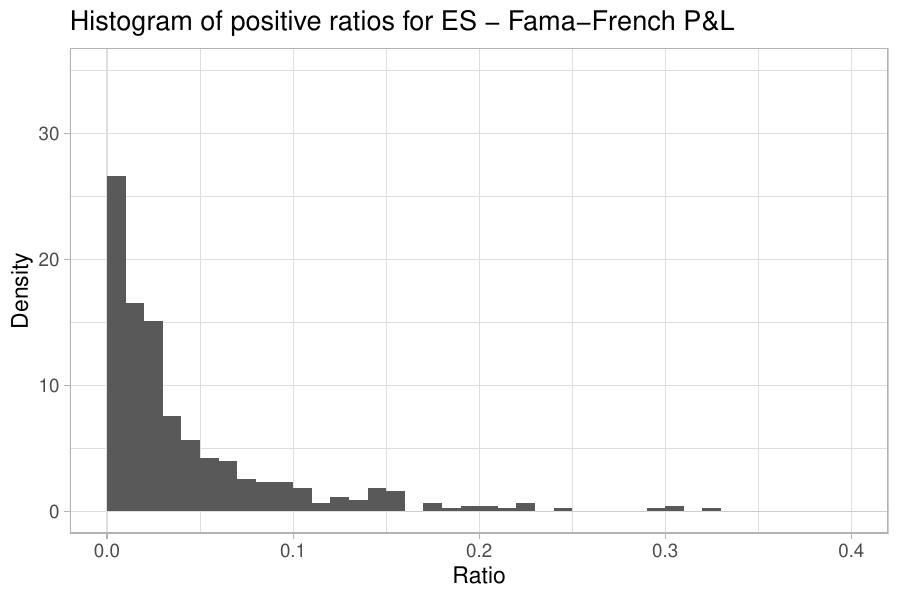}
    \caption{Left panel: the simulated probability of subadditivity violation for ES estimator, discussed in Section~\ref{S:ex2}, with respect to the parameter $\rho$ and with the underlying bivariate Student's $t$-distribution. Middle panel:  the histogram of ES ratios for the bivariate Student's $t$-distribution data with $\rho=0.9$. Right panel: an analogous histogram for Fama--French dataset.}
    \label{fig:ex:2}
\end{figure}

As in Section~\ref{S:ex1}, we empirically check the probability of subadditivity violations for $\widehat\ES^{\mathrm{GPD}}_{\alpha,n}$ using both simulated and market data, and considering the subadditivity ratio $r_{\ES}(\mathbf{x}, \mathbf{x}')$ defined analogously to \eqref{eq:VaR_ratio}. 

For simulated data, to take into account the fact that the GPD estimator is used typically for data with heavy-tails, to get $(\mathbf{x}, \mathbf{x}')$ we use the bivariate Student's $t$ distribution with $\nu=4$ degrees of freedom. As before, we perform MC simulations of size $N=10\,000$ and consider a bivariate Student's t distribution with zero locations, unit scales, and correlation in the same range, that is, $\rho\in\{0, 0.1, \ldots 0.9\}$. The results are presented in Figure~\ref{fig:ex:2}, in a manner analogous to Figure~\ref{fig:ex:1}.

For $\rho=0.9$, the subadditivity is violated in more than 10\% of the runs. Moreover, for 8\% of the violations, estimated ES for $\mathbf{x}+\mathbf{x}'$ exceeded the sum of individual ES estimates by more than 10\%. In particular, this confirms that the GPD plug-in ES estimator is not coherent.

As before, we also checked the subadditivity violation on Fama--French dataset using 2025 daily portfolio return pairs. The subadditivity was violated in 8.7\% cases.\footnote{Out of 4950 pairs of portfolios, the numerical maximum likelihood fit failed in 99 cases, which were excluded from further analysis.} From the right panel of Figure~\ref{fig:ex:2} we notice that the ES estimate for the sum of the two return samples exceeded its upper bound by more than 10\% in 12\% of the violations.

To further investigate this phenomenon, we consider an exemplary pair of portfolios of the Fama--French dataset (ME4.BM10 and ME7.BM8), where the subadditivity was violated, see Figure~\ref{fig:ex:2.1}. 
\begin{figure}[htp!]
    \centering
    \includegraphics[width=0.33\linewidth]{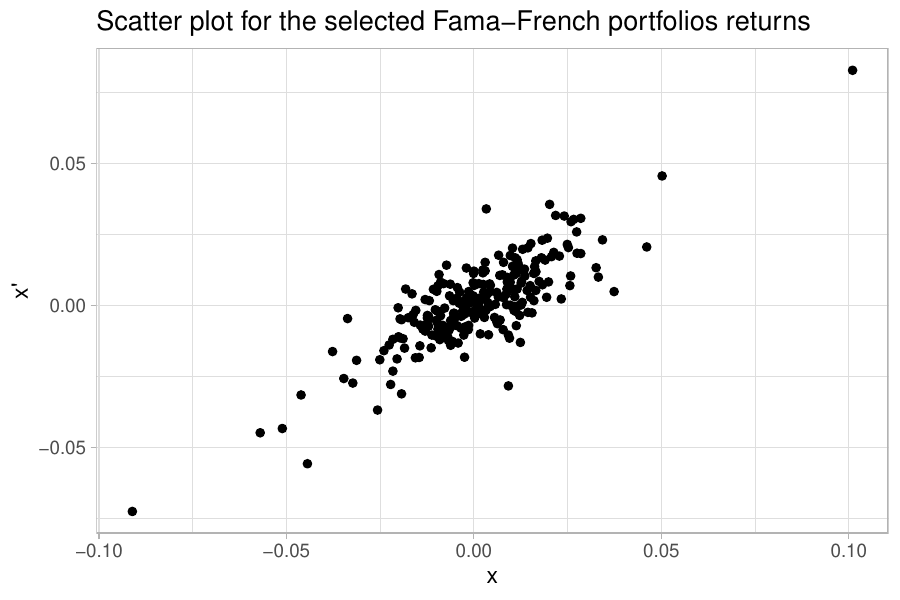}
    \includegraphics[width=0.33\linewidth]{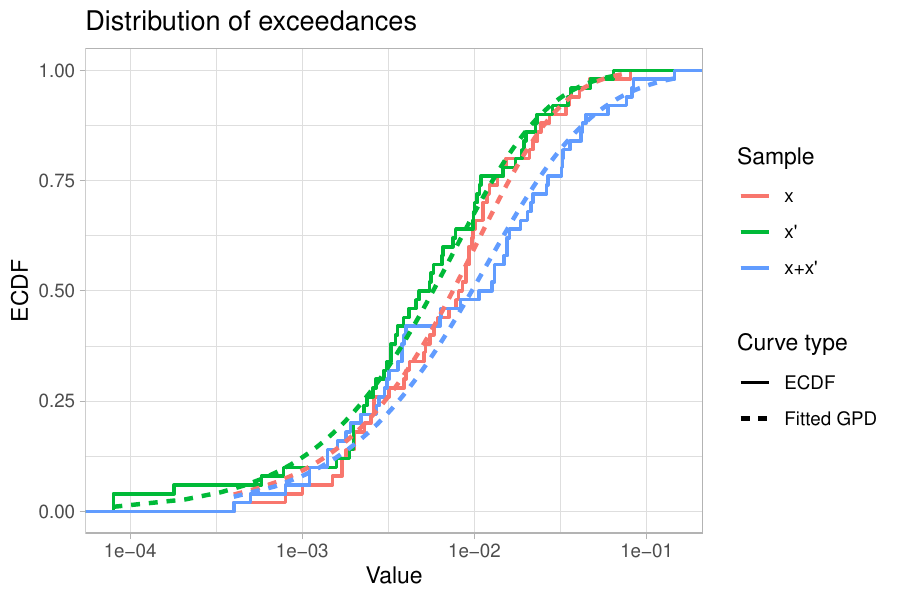}
    \caption{Left panel: the scatter plot of returns for ME4.BM10 (X axis) and ME7.BM8 (Y axis) portfolio from Fama-French dataset. Right panel:  the empirical distribution functions of the corresponding exceedance vectors (solid lines) and the cumulative distribution functions of the fitted GPDs (dashed lines).}
    \label{fig:ex:2.1}
\end{figure}
The left panel shows the scatter plot, indicating a relatively strong correlation between these two portfolios (correlation coefficient equals 0.82). The right panel shows the empirical cumulative distribution functions for the underlying exceedance vectors as well as the CDFs of the fitted GPDs, confirming a relatively good fit for the tails.  Nevertheless, subadditivity of ES is violated: the estimated ES values for ME4.BM10 and ME7.BM8 are $0.053$ and $0.048$, respectively, while the ES of their combined portfolio equals $0.124$, yielding $0.124 > 0.101 = 0.053 + 0.048$. 

In conclusion, we note that the subadditivity property can be substantially violated for non-coherent estimators, even when the underlying risk measure is coherent. This further underscores the importance of the representation results in Section~\ref{sec:robust}, which show that $L$-estimators provide the appropriate structural form for coherent estimation of ES.

\subsection{Weight-induced performance differences among coherent ES estimators}\label{S:ex3}

In this section we provide an illustrative comparison of CRE-based ES estimators, focusing on how different L-estimator weight structures affect the accuracy, bias, and capital adequacy. We compare three ES estimators under a regulatory-style setup with confidence level $\alpha=2.5\%$ and sample size $n=250$, following an i.i.d.\ sampling framework. The first estimator is $\widehat{\ES}^{\textrm{emp}}_{\alpha,n}$, that is, the empirical risk estimator for ES defined in~\eqref{eq:ES_est_FRTB}. 
The second estimator is the popular ES estimator based on the sample conditional mean, that is $\widehat{\ES}^{\textrm{scm}}_{\alpha,n}$, see~\eqref{eq:ES-param}. The third estimator is given by
\[
\textstyle \widehat\ES_{\alpha,n}^{\textrm{qt6}}(\bfx):=\frac{-1}{\alpha(n+1)}\Big(\frac{3}{2} x_{1:n}+\sum_{i=2}^{M_6-1} x_{i:n}+\frac{1+2 R_6-R_6^2}{2} x_{M _6:n}+\frac{R_6^2}{2} x_{(M_6+1):n}\Big),
\]
where $M_6:=\lfloor \alpha(n+1)\rfloor $, $R_6:=\alpha(n+1)-\lfloor\alpha(n+1)\rfloor$, and corresponds to a plug-in estimator associated with the Type~6 sample quantile with flat tail extrapolation; see~\cite[Annex~I]{EBA2023RTS05} and~\cite{HynFan1996}. For notational simplicity, instead of using full names, we refer to these three estimators simply as \#1, \#2, and \#3, respectively, see Table~\ref{T:weights} where we also report exact weighting schemes and total mass. Our focus is on how differences in the associated order statistic weights influence estimation accuracy and related performance characteristics.
\begin{table}[htp!]
    \centering
    \begin{tabular}{cccccccccc}
        Id & Estimator & $a_1$ & $a_2$ & $a_3$ & $a_4$ & $a_5$ & $a_6$ & $a_7$ & sum\\\hline \noalign{\vspace{2pt}}
        \#1 & $\widehat{\ES}^{\textrm{emp}}_{\alpha,n}$ & 0.160 & 0.160 & 0.160 & 0.160 & 0.160 & 0.160 & 0.040& 1.000\\[2pt]
        \#2 & $\widehat{\ES}^{\textrm{scm}}_{\alpha,n}$ & 0.167 & 0.167 &  0.167 & 0.167 & 0.167 & 0.167 & 0.000 & 1.000 \\[2pt]
        \#3 & $\widehat{\ES}^{\textrm{qt6}}_{\alpha,n}$ & 0.239 & 0.159 & 0.159 & 0.159 & 0.159  & 0.117 & 0.006& 1.000
    \end{tabular}
    \caption{The weights assigned to the first seven order statistics for estimator \#1-\#3, calculated for $\alpha=2.5\%$ and $n=250$, with  the remaining weights equal to 0. The weights for all estimators are decreasing, which is consistent with CRE representation \eqref{eq:comonotonic_representation}. The results are rounded to 3 decimal digits.}
    \label{T:weights}
\end{table}

In view of Theorem~\ref{th:law_inv_est_representation}, all three estimators satisfy the CRE representation, since their weight vectors are decreasing and sum to one. The weights may also be interpreted as alternative discrete approximations of risk spectra, cf.~Example~\ref{ex:consistenty.srm}. Despite common structure, the estimators differ markedly in how weights are allocated across order statistics, particularly at the boundaries (e.g.\ $a_1$ and $a_6$ or $a_7$). While such differences are unlikely to affect bias under moderately tailed distributions, they may induce systematic distortions in the presence of heavy tails or when extreme observations arise from exogenous shocks.

To quantify these effects, we evaluate estimator performance using a set of benchmark metrics,  summarized in Table~\ref{T:performance.metrics}, commonly employed in regulatory and simulation-based studies, following the framework of~\cite[Annex~I]{EBA2023RTS05}.

\begin{table}[htp!]
\centering
\small
\setlength{\tabcolsep}{4pt}
\begin{tabular}{p{0.9cm} p{4.2cm} p{4.6cm} m{5.8cm}}
\toprule
Metric 
& Theoretical definition 
& Monte Carlo implementation 
& Interpretation \\ 
\midrule
\textbf{SE} 
&
$\displaystyle
\frac{\sqrt{\mathbb{E}\!\left[(\widehat{\ES}_{\alpha,n}-\ES_\alpha)^2\right]}}
{\ES_\alpha}
$
&
$\displaystyle
\frac{\sqrt{\tfrac{1}{K}\sum_{k=1}^K
(\widehat{\ES}_{\alpha,n}(\bfx_k)-\ES_\alpha)^2}}
{\ES_\alpha}
$
&
The (scaled) root mean squared error is an accuracy measure. Lower values indicate tighter concentration around the true ES. \\[1.2ex]
\midrule
\textbf{BI} 
&
$\displaystyle
\frac{\mathbb{E}[\widehat{\ES}_{\alpha,n}]-\ES_\alpha}
{\ES_\alpha}
$
&
$\displaystyle
\frac{1}{K}\sum_{k=1}^K
\frac{\widehat{\ES}_{\alpha,n}(\bfx_k)}{\ES_\alpha}-1
$
&
Statistical Bias measures systematic over- or underestimation. Values close to zero indicate approximate unbiasedness. \\[1.2ex]
\midrule
\textbf{CT} 
&
$\displaystyle
\inf_{\beta\in(0,1)}
\{\ES_\beta(X+\widehat{\ES}_{\alpha,n})\ge 0\}
$
&
$\displaystyle
\inf_{\beta\in(0,1)}
\Bigl\{
\widehat{\ES}^{\textrm{emp}}_{\beta,K}
(\tilde x_k+\widehat{\ES}_{\alpha,n}(\bfx_k))\ge 0
\Bigr\}
$
&
Safe Confidence Threshold measures the minimal confidence level at which the secured position is acceptable. Deviation from $\alpha$ reflects capital adequacy.\footnote{This measure can be viewed as an acceptability index (or performance measure) dual to the ES family of risk measures, that verifies whether the estimated capital reserve secures the portfolio at a confidence level close to the reference value $\alpha \in (0,1)$. We refer to \cite{MolPit2017}, where this metric is discussed in details, and used to construct a targeted ES backtest.}
\\
\bottomrule
\end{tabular}
\caption{Benchmark performance metrics for ES estimator comparison. SE and BI are reported in relative terms (percentages of the true ES). Monte Carlo estimates are computed using $K=10^7$ replications. Here $\bfx_k$ refers to $k$th MC sample and $\tilde x_k$ refers to the $k$th independent pick from the underlying distribution.}
\label{T:performance.metrics}
\end{table}

In the numerical study, benchmark performance metrics are computed using Monte Carlo simulation for a representative family of distributions following \cite[Annex~I]{EBA2023RTS05}. Specifically, we consider eight distributions: the standard normal distribution, a Student’s $t$-distribution with $\nu=5$ degrees of freedom, and six normal–inverse Gaussian (NIG) distributions with parameters
\[
(\alpha,\beta)\in\{(0.4,\pm0.14),(0.55,\pm0.3025),(0.4,\pm0.22)\},
\]
and with location $\mu=0$ and scale $\delta=1$. The NIG family allows for flexible control of skewness and kurtosis, which are available in closed form and preserved under convolution. For the chosen parameter values, the corresponding skewness–excess kurtosis pairs range from $(\pm0.9460,\,3.6282)$ to $(\pm1.8702,\,8.5174)$, covering magnitudes compatible to stress period trading-book P\&Ls, based on banks' real trading book P\&L data from the period 2014–2022~\cite[Annex~I, p.~162]{EBA2023RTS05}, cf. {also}~\cite{The2020}. 

For each distribution, all performance metrics are evaluated using Monte Carlo simulations with $K=10^7$ replications. In addition to results for the ES$_{2.5\%}$ estimators~\#1–\#3, we also report benchmark values for the value-at-risk 
estimator based on the linearly interpolated Type~6 sample quantile for $n=250$ and $\alpha=1\%$, that is,
\begin{equation}\label{eq:var.interpolation}
\widehat\var_{1\%}(\bfx):=-\left( 0.49 x_{(2:250)} + 0.51 x_{(3:250)}\right),
\end{equation}
see ~\cite{HynFan1996} and~\cite[p. 267]{EGIM}. 
Note that $\var_{1\%}$ is a reference market risk metric in the Basel II framework and for the normally distributed random variable $X$ we get $\var_{1\%}(X) = 2.326$ and $\ES_{2.5\%}(X)= 2.336$, cf.~\eqref{eq:es-normal}, so that this metric could be used for VaR and ES comparison purposes; see also  \cite{KelRos2016}.

The numerical results presented in Table~\ref{tab:biases:RTS} reveal systematic differences in performance across the considered ES estimators.
\begin{table}[htp!]
    \centering
\scalebox{0.8}{
\centering
\small
\setlength{\tabcolsep}{4pt}

\begin{minipage}[t]{0.28\textwidth}
\centering
\caption*{\textbf{Normal}}
\begin{tabular}{lccc}
\toprule
Est. & $SE$ & $BI$ & $CT$ \\
\midrule
\#1 & \textbf{8.7\%} & -1.5\% & 3.1\% \\
\#2 & 8.7\% & \textbf{-0.8\%} & 3.0\% \\
\#3 & 9.2\% & 1.3\% & \textbf{2.6\%} \\
\midrule
\gray{$\widehat{\mathrm{VaR}}_{1\%}$}
& \gray{10.8\%} & \gray{3.3\%} & \gray{1.0\%} \\
\bottomrule
\end{tabular}
\end{minipage}
\hspace{0.02\textwidth}
\begin{minipage}[t]{0.28\textwidth}
\centering
\caption*{\textbf{Student's $t$ ($\nu=5$)}}
\begin{tabular}{lccc}
\toprule
Est. & $SE$ & $BI$ & $CT$ \\
\midrule
\#1 & \textbf{16.7\%} & -1.8\% & 3.2\% \\
\#2 & 17.1\% & \textbf{-0.8\%} & 3.1\% \\
\#3 & 19.9\% & 3.6\% & \textbf{2.8\%} \\
\midrule
\gray{$\widehat{\mathrm{VaR}}_{1\%}$}
& \gray{21.6\%} & \gray{7.9\%} & \gray{1.0\%} \\
\bottomrule
\end{tabular}
\end{minipage}
\hspace{0.02\textwidth}
\begin{minipage}[t]{0.28\textwidth}
\centering
\caption*{\textbf{NIG ($\alpha=0.4,\ \beta=-0.22$)}}
\begin{tabular}{lccc}
\toprule
Est. & $SE$ & $BI$ & $CT$ \\
\midrule
\#1 & \textbf{20.4\%} & -2.2\% & 3.2\% \\
\#2 & 20.8\% & \textbf{-0.8\%} & 3.1\% \\
\#3 & 23.7\% & 4.7\% & \textbf{2.7\%} \\
\midrule
\gray{$\widehat{\mathrm{VaR}}_{1\%}$}
& \gray{28.1\%} & \gray{10.4\%} & \gray{1.0\%} \\
\bottomrule
\end{tabular}
\end{minipage}
\hspace{0.02\textwidth}
\begin{minipage}[t]{0.28\textwidth}
\centering
\caption*{\textbf{NIG ($\alpha=0.4,\ \beta=0.22$)}}
\begin{tabular}{lccc}
\toprule
Est. & $SE$ & $BI$ & $CT$ \\
\midrule
\#1 & \textbf{17.8\%} & -2.1\% & 3.2\% \\
\#2 & 18.2\% & \textbf{-0.9\%} & 3.1\% \\
\#3 & 20.4\% & 3.9\% & \textbf{2.7\%} \\
\midrule
\gray{$\widehat{\mathrm{VaR}}_{1\%}$}
& \gray{24.1\%} & \gray{8.7\%} & \gray{1.0\%} \\
\bottomrule
\end{tabular}
\end{minipage}
}

\par\medskip

\scalebox{0.8}{
\centering
\small
\setlength{\tabcolsep}{4pt}

\begin{minipage}[t]{0.28\textwidth}
\centering
\caption*{\textbf{NIG ($\alpha=0.4,\ \beta=0.14$)}}
\begin{tabular}{lccc}
\toprule
Est. & $SE$ & $BI$ & $CT$ \\
\midrule
\#1 & \textbf{17.9\%} & -2.1\% & 3.2\% \\
\#2 & 18.2\% & \textbf{-0.9\%} & 3.1\% \\
\#3 & 20.5\% & 3.9\% & \textbf{2.7\%} \\
\midrule
\gray{$\widehat{\mathrm{VaR}}_{1\%}$}
& \gray{24.2\%} & \gray{8.8\%} & \gray{1.0\%} \\
\bottomrule
\end{tabular}
\end{minipage}
\hspace{0.02\textwidth}
\begin{minipage}[t]{0.28\textwidth}
\centering
\caption*{\textbf{NIG ($\alpha=0.4,\ \beta=-0.14$)}}
\begin{tabular}{lccc}
\toprule
Est. & $SE$ & $BI$ & $CT$ \\
\midrule
\#1 & \textbf{19.3\%} & -2.1\% & 3.2\% \\
\#2 & 19.7\% & \textbf{-0.8\%} & 3.1\% \\
\#3 & 22.3\% & 4.4\% & \textbf{2.7\%} \\
\midrule
\gray{$\widehat{\mathrm{VaR}}_{1\%}$}
& \gray{26.5\%} & \gray{9.7\%} & \gray{1.0\%} \\
\bottomrule
\end{tabular}
\end{minipage}
\hspace{0.02\textwidth}
\begin{minipage}[t]{0.28\textwidth}
\centering
\caption*{\textbf{NIG ($\alpha=0.55,\ \beta=0.3025$)}}
\begin{tabular}{lccc}
\toprule
Est. & $SE$ & $BI$ & $CT$ \\
\midrule
\#1 & \textbf{16.8\%} & -2.1\% & 3.2\% \\
\#2 & 17.1\% & \textbf{-0.9\%} & 3.1\% \\
\#3 & 19.2\% & 3.5\% & \textbf{2.7\%} \\
\midrule
\gray{$\widehat{\mathrm{VaR}}_{1\%}$}
& \gray{22.6\%} & \gray{8.0\%} & \gray{1.0\%} \\
\bottomrule
\end{tabular}
\end{minipage}
\hspace{0.02\textwidth}
\begin{minipage}[t]{0.28\textwidth}
\centering
\caption*{\textbf{NIG ($\alpha=0.55,\ \beta=-0.3025$)}}
\begin{tabular}{lccc}
\toprule
Est. & $SE$ & $BI$ & $CT$ \\
\midrule
\#1 & \textbf{18.6\%} & -2.1\% & 3.2\% \\
\#2 & 19.0\% & \textbf{-0.8\%} & 3.1\% \\
\#3 & 21.5\% & 4.1\% & \textbf{2.7\%} \\
\midrule
\gray{$\widehat{\mathrm{VaR}}_{1\%}$}
& \gray{25.4\%} & \gray{9.2\%} & \gray{1.0\%} \\
\bottomrule
\end{tabular}
\end{minipage}
    }
    \caption{Performance output for ES estimators \#1-\#3. Sample size $n=250$ and confidence threshold $\alpha=2.5\%$. 
    Performance metrics are defined in Table~\ref{T:performance.metrics}, estimator $\widehat\var_{1\%}$ is given in \eqref{eq:var.interpolation}. The ES estimators with the best performance are highlighted in bold.}
    \label{tab:biases:RTS}
\end{table}
The empirical estimator~\#1 consistently achieves the smallest scaled RMSE across all distributions, thereby providing the best overall fit among the ES candidates. Estimator~\#2, in turn, exhibits the smallest absolute statistical bias throughout the considered scenarios. Hence, while estimator~\#1 offers superior estimation accuracy in terms of variability-adjusted error, estimator~\#2 delivers more accurate centering relative to the true ES; the bias is also less negative when confronted with \#1, which leads to improved control over potential risk underestimation. The ranking changes when confidence-threshold performance is considered. In this dimension, estimators~\#1 and~\#2 are systematically outperformed by estimator~\#3, which produces more conservative risk estimates and achieves improved control of the exceedance frequency. The improved coverage of estimator~\#3, however, comes at the cost of a positive (conservative) bias and a deterioration in scaled RMSE relative to the other ES estimators. 

The SE of the ES$_{2.5\%}$ estimators is consistently lower than SE of the $\var_{1\%}$ estimator, i.e. at comparable capital level the fit for ES is better than the fit for $\var$. However, CT performance appears more sensitive to the underlying distributional specification.

These findings are consistent with the weight structures reported in Table~\ref{T:weights}. The nearly uniform weighting scheme of estimator~\#2 mitigates statistical bias but slightly reduces estimation efficiency due to the exclusion of the seventh order statistic. By contrast, estimator~\#1 achieves improved RMSE performance through a more balanced allocation of weights across tail observations. In both estimators~\#1 and~\#2, the relatively small weight assigned to the most extreme losses appears insufficient to ensure strong capital protection, whereas the increased tail emphasis in estimator~\#3 enhances confidence coverage at the expense of bias and overall fit.

 To further verify our observations, let us consider the Fama--French dataset. While one cannot meaningfully evaluate SE and BI metrics without knowing the true value of the underlying risk, the CT metric could be evaluated on market data following the backtesting scheme described in \cite{MolPit2017}. In a nutshell, for each portfolio from the dataset, and each daily return in 2024-2025, we estimate ES using the prescribed methodologies on the preceding $n=250$ observations and then construct the secured portfolio sample of size $y=(y_i)_{i=1}^{B}$, where $B=500$ is the number of backtesting observations corresponding to the number of daily returns in the period 2024-2025; secured sample is the sum of realized value and estimated risk, as in the definition of CT metric. Then, we compute the backtesting metric $\overline{CT}:=\tfrac{1}{B}\sum_{i=1}^{B}\1_{\{y_1\ldots+y_i<0\}}$ which is the analogue of the MC implementation of the CT metric; see \cite[Proposition 3.1]{MolPit2017} for details. For estimators \#1, \#2, \#3, we get the values of $\overline{CT}$ equal to  4.8\%, 4.6\%, and 4.3\%, respectively; for VaR, the exception rate is equal to 1.5\%. While the ordering of $\overline{CT}$ values is consistent with the ordering of CT values, we note that all metrics materially increased -- this is in fact also true for the VaR metric, showing that one should be careful when applying an i.i.d.-based estimator to empirical data.

More broadly, the results show that estimation accuracy, bias, and capital protection can be systematically adjusted through the choice of weights, motivating further research on robust CRE representations. While the present study is stylized, practical settings often involve additional features such as heteroskedasticity or overlapping P\&L samples, which can substantially affect estimator performance. Comparisons with~\cite[Annex~I]{EBA2023RTS05} and~\cite{AicCroReh2021} confirm that performance may deteriorate when overlapping data are used. This highlights the importance of careful estimator selection and calls for further systematic study of overlapping constructions, which remain insufficiently explored despite their regulatory relevance, cf.~\cite[MAR 33.4(7)]{Bas2019}.

\section{\rev{Concluding remarks}}\label{S:concluding.remarks}
We develop a theoretically sound framework for the estimation of CRM, inspired by the axiomatic theory of risk measures, by imposing analogous properties on estimators as functions of the sampled P\&L, mirroring those satisfied by coherent risk measures as functionals of the population law. The obtained theoretical results fully characterize these risk estimators, that we call CREs.  In particular, we show that any law-invariant CRE admits a representation as a supremum over a family of L-statistics, that is, as the supremum over linear transformations of the sorted sample. Moreover, any law-invariant comonotonic CRE reduces to a single L-estimator, giving a computationally desirable representation.  
These results imply that some commonly used estimators are not in the class of CREs; this includes most parametric plug-in methods that are shown to be structurally incompatible with estimation coherency.

Through a series of examples, we demonstrate that coherence of the underlying (theoretical) risk measure does not automatically extend to its estimators, and that violations of coherence occur with non-negligible probability and are not mere artifacts. In particular, using simulated and market data, we show that: parametric plug-in estimators for Gaussian distributions may fail monotonicity; semi-parametric ES estimators for the GPD are not subadditive with positive probability; sampling from elliptical distributions can lead to subadditivity violations in estimators, even when the original risk measure is subadditive on population level -- for example, this applies to the empirical $\var$, one of the most widely used risk estimators. More generally, parametric (plug-in) estimators possess a fundamentally different structure from CREs. On the other hand, we prove that any empirical risk estimator derived from a law-invariant CRM is itself a law-invariant CRE, which provides a natural and canonical construction of CREs.

The obtained results show that order statistics and $L$-estimators are not merely convenient or incidental: they are structurally necessary for constructing estimators that satisfy important conceptual properties. Such properties are also important in view of current FRTB regulation, see e.g.~\cite[Article 42]{EU2024_1085}. The choice of $L$-statistic weights is a key design element, as illustrated in Section~\ref{S:ex3}, and can substantially affect outcomes in terms of accuracy, sampling error, and estimation confidence even when the risk estimator is coherent. 

Finally, we note that while this study primarily focuses on i.i.d. case and general estimator construction logic, the core definitions and main results, including Theorem~\ref{th:coherent_estimator_representation}, remain valid in a general setting. Preliminary numerical evidence suggests that the main conclusions extend to non-i.i.d. data; however, as expected, theoretical developments in this broader context require new methods and techniques. This is especially important when our results are compared with results presented in~\cite[Annex~I]{EBA2023RTS05} and~\cite{AicCroReh2021} which indicate that the estimation performance might substantially deteriorate when overlapping or heteroskedastic data are used. This highlights the importance of careful estimator selection and calls for further systematic study of overlapping constructions which are left for future work.


\section*{Disclaimer}
\noindent The views and opinions expressed in this paper are the authors’ own and do not necessarily reflect the views and opinions of their current or past employers. In particular,  they cannot be taken to represent those of the European Central Bank (ECB) or to state the ECB’s policy. Neither the ECB nor any person acting on its behalf may be held responsible for the use which may be made of the information contained in this publication, or for any errors which, despite careful preparation and checking, may appear therein.

\section*{Funding}
\noindent Igor Cialenco acknowledges support from the US National Science Foundation grant DMS-2407549.  Damian Jelito acknowledges support from the National Science Centre, Poland, via project 2024/53/B/ST1/00703. Marcin Pitera acknowledges support from the National Science Centre, Poland, via project 2024/53/B/HS4/00433.

\section*{Acknowledgements}
\noindent Martin Aichele thanks Carlo Acerbi for helpful discussions on preparatory work on ES estimation.

\theendnotes

\bibliography{bibliography}

\end{document}